\documentclass[journal, twoside, final]{IEEEtran}

\usepackage{graphicx}
\usepackage[cmex10]{amsmath}
\usepackage{cases}
\usepackage[tight,footnotesize]{subfigure}
\usepackage{amsthm} %proof
\usepackage{amsmath}
\usepackage{cite}
\usepackage{amssymb}
\usepackage{algorithm}
\usepackage{algorithmic}
\usepackage{xcolor}
\usepackage{stfloats}
\usepackage{multirow}
\usepackage{CJK}
\usepackage{subeqnarray}
\usepackage{longtable}
\usepackage{multicol}
\usepackage{extpfeil}
\usepackage{verbatim}
\newtheorem{lemma}{\bf{Lemma}}%[section]
\newtheorem{theorem}{\bf{Theorem}}

\newtheorem{prop}{\bf{Proposition}}

\begin{document}

\title{Massive Access in Secure NOMA under Imperfect CSI: Security Guaranteed Sum-Rate Maximization with First-Order Algorithm}

\author{Zongze Li, Minghua Xia, Miaowen Wen and Yik-Chung Wu

\thanks{Manuscript received January 27, 2020; revised May 31, 2020; accepted July 17, 2020. This work was supported in part by the National Natural Science Foundation of China under Grant 61671488, in part by the Major Science and Technology Special Project of Guangdong Province under Grant 2018B010114001, and in part by the Fundamental Research Funds for the Central Universities under Grants 191gjc04 and 2019SJ02.
	
Zongze Li and Yik-Chung Wu are with the Department of Electrical and Electronic Engineering, The University of Hong Kong, Hong Kong (e-mail: \{zzli,ycwu\}@eee.hku.hk).
    
Minghua Xia is with the School of Electronics and Information Technology, Sun Yat-sen University, Guangzhou 510006, China, and also with the Southern Marine Science and Engineering Guangdong Laboratory, Zhuhai 519082, China (e-mail: xiamingh@mail.sysu.edu.cn).
    
Miaowen Wen is with the School of Electronics and Information Engineering, South China University of Technology, Guangzhou 510640, China (e-mail: eemwwen@scut.edu.cn).
}

}

% make the title area
\maketitle

\begin{abstract}
Non-orthogonal multiple access (NOMA) is a promising solution for secure transmission under massive access. However, in addition to the uncertain channel state information (CSI) of the eavesdroppers due to their passive nature, the CSI of the legitimate users may also be imperfect at the base station due to the limited feedback. Under both channel uncertainties, the optimal power allocation and transmission rate design for a secure NOMA scheme is currently not known due to the difficulty of handling the probabilistic constraints. This paper fills this gap by proposing novel transformation of the probabilistic constraints and variable decoupling so that the security guaranteed sum-rate maximization problem can be solved by alternatively executing branch-and-bound method and difference of convex programming. To scale the solution to a truly massive access scenario, a first-order algorithm with very low complexity is further proposed. Simulation results show that the proposed first-order algorithm achieves identical performance to the conventional method but saves at least two orders of magnitude in computation time. Moreover, the resultant transmission scheme significantly improves the security guaranteed sum-rate compared to the orthogonal multiple access transmission and NOMA ignoring CSI uncertainty.
\end{abstract}

\begin{IEEEkeywords}
First-order algorithm, limited feedback, massive access, non-orthogonal multiple access, outage probability, physical layer security.
\end{IEEEkeywords}

\IEEEpeerreviewmaketitle

\section{Introduction}

With the explosive growth of the Internet-of-Things, massive number of users and devices will access wireless networks at the same time~\cite{B_19WireFuture,J_LiuMA18_R2,J_Zhang19prospective}. 
However, the large amount of data flowing in various wireless propagation channels poses significant privacy and security challenges to the next-generation wireless system design~\cite{J_19SurveyFutureNOMA}. 
Although physical layer security exploits channel capacity difference between the legitimate user and the eavesdropper to protect the legitimate transmission, this technique, when applied in its primitive form, is vulnerable in multiple access systems since eavesdroppers have more targets to choose from~\cite{J_Wu18Survey}. 
Fortunately, non-orthogonal multiple access (NOMA) technique provides a promising solution by serving a group of legitimate users through power domain multiplexing~\cite{J_Wei_DN20}, thus generating artificial interference to the eavesdroppers~\cite{J_SeNOMA19_Zeng}.

Pioneering works on NOMA with security consideration assume perfect channel state information (CSI) of the legitimate users' channel and imperfect CSI of the eavesdropper's channel~\cite{J_Sun18NOMA,J_Feng19SeNoma,J_YueSecure20NOMA}.
Under this assumption, secure transmission with NOMA achieves a higher
secure transmission rate (i.e., transmission rate with secrecy outage probability constraint satisfied) than the time-division multiplexing and frequency-division multiplexing access~\cite{J_Wang19SecureNOMI}.
While these results are encouraging, the assumption on perfect CSI of the legitimate users' channels is too strong in practice, especially with massive users where a large number of CSIs from legitimate users need to be fed back to the base station (BS). As a result, limited feedback or vector quantization error is inevitable. Due to the uncertainty in the legitimate users' CSIs, an outage may also occur during the legitimate transmission~\cite{J_Wei17OptPower}.

To ensure both the legitimate transmission outage probability and the secrecy outage probability are within tolerable levels, the secure transmission scheme design should incorporate two probabilistic constraints, which unfortunately do not admit closed-form expressions for further analysis.
To overcome this challenge, this paper adopts the quantization cell approximation~\cite{J_Yoo_pdf07} and the Bernstein-type inequality~\cite{J_Bech09Berntein} to transform the intractable probabilistic constraints into deterministic ones.
Furthermore, by introducing an auxiliary variable and leveraging Lambert W function, 
the constraints are decoupled. Then, the resultant problem is readily solved via block coordinate ascent approach~\cite{C_Xie18BCA}, where branch-and-bound algorithm~\cite{J_Mult_FP01} and difference of convex (DC) programming~\cite{J_An05Opt} are respectively used to handle each subproblem.

While the above algorithm is a workable solution, it does not scale well with network size. When the size of the network is large, both branch-and-bound algorithm and DC programming would be too time-consuming.
To scale the transmission scheme to massive access, a fast first-order algorithm is further proposed by exploiting the block separability of constraints under the alternating maximization framework~\cite{Bert1989Padc}. In particular, to replace the branch-and-bound algorithm, the multiple-ratio structure of the objective function is exploited, and a quadratic transform is employed to obtain an iterative algorithm, which converges to a local optimal solution.
On the other hand, to get around the DC programming, a projected-gradient method~\cite{B_Beck09_AM} is employed to obtain a stationary point. 
It is proved that the overall first-order algorithm is guaranteed to converge. Furthermore, simulation results demonstrate that the proposed first-order algorithm reduces computation time by at least two orders of magnitude compared to the method employing branch-and-bound and DC programming while achieving the same performance. 
Finally, simulation results show that the resultant transmission scheme achieves a significantly higher security guaranteed sum-rate than the orthogonal multiple access scheme and NOMA scheme ignoring CSI uncertainty.

The rest of this paper is organized as follows. System model and the security guaranteed sum-rate maximization problem are formulated in Section~\ref{Sec:II}. Then, the probabilistic constraints are transformed into deterministic constraints in Section~\ref{Sec:III}. In  Sections~\ref{Sec:IV} and~\ref{Sec:V}, a conventional solution and a first-order algorithm are respectively proposed. Simulation results are presented in Section~\ref{Sec:VI}. Finally, conclusion is drawn in Section~\ref{Sec:VII}.

$\mathit{Notation:}$ Column vectors and matrices are denoted by lowercase and uppercase boldface letters, respectively. Conjugate transpose, transpose, Frobenius norm, trace, the modulus of a scalar and mathematical expectation are denoted by $(\cdot)^H$, $(\cdot)^T$, $\|\cdot\|_F$, $\mathrm{Tr}(\cdot)$, $|\cdot|$ and $\mathbb{E}\{\cdot\}$, respectively. The notations $[x]^+$ and  $\mathrm{Pr}(\cdot)$ stand for $\max\{x,0\}$ and probability, respectively. 
$\mathcal{CN}\left(0,b\right)$ denotes the circularly symmetric complex normal distribution with zero mean and variance $b$, and $\mathrm{Exp}(\lambda)$ denotes the exponential distribution with mean $\lambda$. The principal branch of Lambert W function is defined by $W_0(x)e^{W_0(x)}=x$ for $x\geq-1/e$ with $W_0(x)\geq-1$~\cite{Lambert_W96}.

\section{System Model and Problem Formulation}\label{Sec:II}
\begin{figure}[tb]
    \centering
    \includegraphics[scale=0.52]{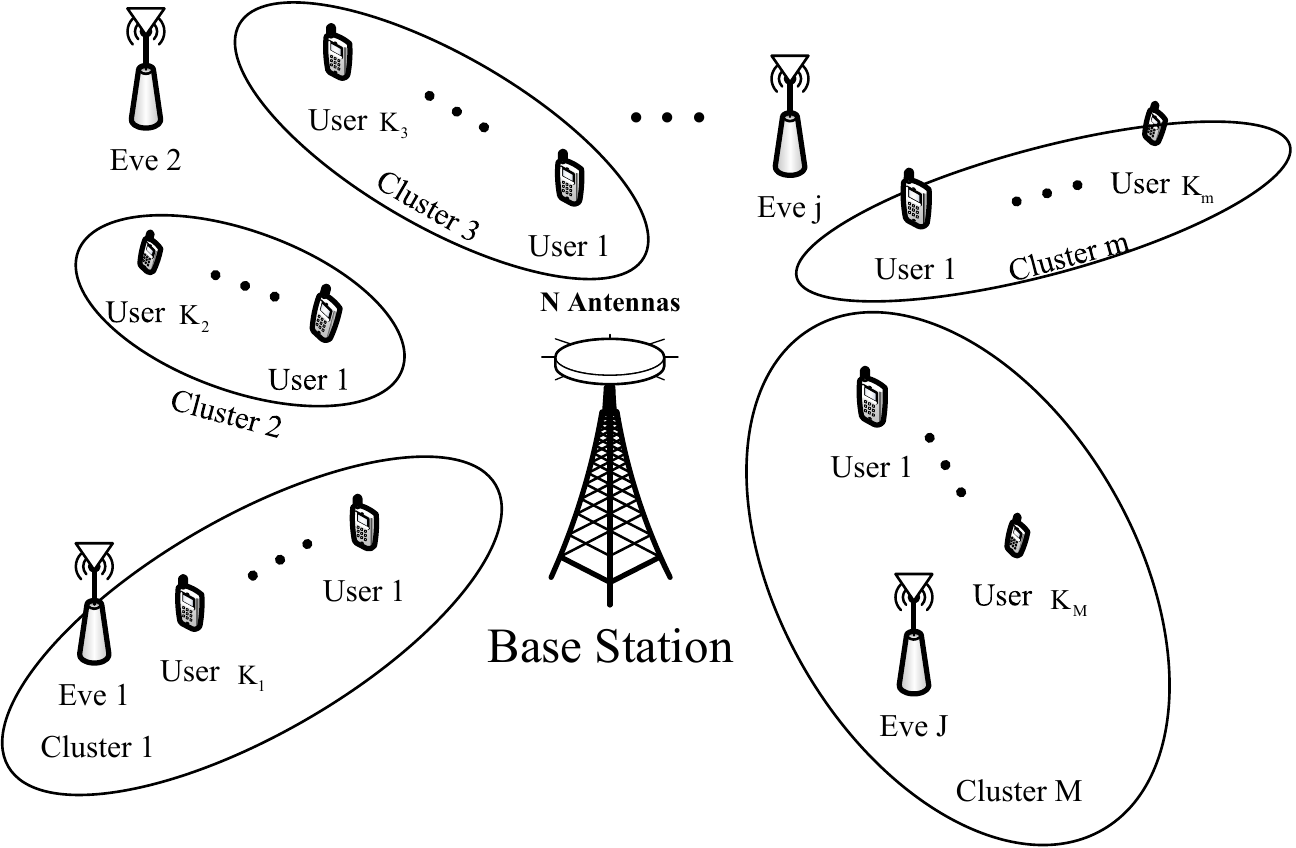}
    \caption{System model of downlink massive access secure NOMA network.}\label{fig:SystemModel}
\end{figure}
We consider a downlink secure multiple-input single-output (MISO) system as shown in Fig.~\ref{fig:SystemModel}. There are one $N$-antennas BS, $M$ clusters of single-antenna legitimate users with cluster $m$ having $K_m$ users, and $J$ passive single-antenna non-colluding eavesdroppers (Eves) who potentially wiretap any message being sent. 
This paper considers the massive access setting in which $\sum_{m=1}^{M} K_m > M$ and assumes that all wireless channels experience quasi-static Rayleigh block fading (i.e., all channels are subject to zero mean complex Gaussian distribution).
The main channel (i.e., from the BS to user $k$ in cluster $m$) and the eavesdropper's channel (i.e., from the BS to Eve $j$) are respectively denoted by $d^{-\alpha/2}_{m,k}\mathbf{g}_{m,k}\in \mathbb{C}^{N\times 1}$ and $d^{-\alpha/2}_{e,j}\mathbf{g}_{e,j}\in \mathbb{C}^{N\times 1}$, where $\mathbf{g}_{m,k}\sim \mathcal{CN}(\mathbf{0},\mu^2_{m,k}\mathbf{I}_N),\mathbf{g}_{e,j}\sim \mathcal{CN}(\mathbf{0},\mu^2_{e,j}\mathbf{I}_N)$ are the small-scale fading vectors, $d_{m,k}$ and $d_{e,j}$ respectively denote the distances from the BS to user $k$ in cluster $m$ and the Eve $j$, and $\alpha$ denotes the path-loss exponent. 
At network initialization, a randomly generated codebook consisting of $M=2^B$ unit-norm vectors each with length $N$ (denoted by $\{\hat{\mathbf{g}}_{m}\}_{m=1}^M$) is designed off-line and made known at both the BS and users via codebook distribution scheme~\cite{B_Clerckx13}.
Then, the BS sends a sequence of training symbols to all users who perform channel estimation to obtain knowledge of their own channels.
Since channel estimation error is negligible when the training sequence is long and when signal-to-noise ratio is high~\cite{J_YCestimation08}, it is assumed that the CSI of the main channel is accurate at the users. After channel quantization (using the codebook obtained earlier from the BS), 
each user conveys its channel direction information (CDI) using $B$ bits over the feedback channel to the BS.
Due to the limited feedback per channel coherence block, the CSI of the main channel obtained at the BS is imperfect~\cite{J_Kout12DownSDMA}.  But the BS can still group users into clusters based on the obtained CDI.

Employing the well-known quantization cell approximation~\cite{J_Jindal06Finite}, the users having the maximum inner product between their small-scale fading vectors and $\hat{\mathbf{g}}_{m}$ are assigned to cluster $m$, and the value of $K_m$ is automatically obtained after the grouping. The normalized channel vector $\tilde{\mathbf{g}}_{m,k}:=\mathbf{g}_{m,k}/\|\mathbf{g}_{m,k}\|$ and $\hat{\mathbf{g}}_{m}$ are related by~\cite{J_Kout12DownSDMA}
\begin{equation}\label{eq:rela_3_gk_channel_imperfect}
\tilde{\mathbf{g}}_{m,k}=(\cos\beta_{m,k} ) \hat{\mathbf{g}}_{m}+(\sin \beta_{m,k}) {\mathbf{e}}_{m,k},
\end{equation}
where ${\mathbf{e}}_{m,k}\in \mathbb{C}^{N\times 1}$ is the unit-norm quantization error vector isotropically distributed in the nullspace of $\hat{\mathbf{g}}_{m}$, and $\beta_{m,k}$ represents the angle between $\tilde{\mathbf{g}}_{m,k}$ and $\hat{\mathbf{g}}_{m}$ with $\sin^2 \beta_{m,k}$ being a random variable with variance determined by $B$~\cite{J_Jindal06Finite}.

By employing NOMA transmission, the superimposed signal $\mathbf{x}$ transmitted by the BS is 
$\mathbf{x}=\sum_{m=1}^{M}\left(\mathbf{w}_m\sum_{k=1}^{K_m}\sqrt{P\theta_{m,k}}s_{m,k}\right)$,
where $\mathbf{w}_m\in \mathbb{C}^{N\times 1}$ is the unit-norm beamforming vector for cluster $m$; $P$ represents the total transmit power; $\theta_{m,k}$ and $s_{m,k}$ respectively denote power allocation ratio and information bearing signals for user $k$ in cluster $m$ with $\mathbb{E}\{|s_{m,k}|^2\}=1$. The received signals at user $k$ in cluster $m$ and Eve $j$ are respectively given by
\begin{equation}\label{eq:user_signal_mk}
y_{m,k}=d^{-\alpha/2}_{m,k}\mathbf{g}^H_{m,k} \sum_{i=1}^{M}\left(\mathbf{w}_i\sum_{v=1}^{K_i}\sqrt{P\theta_{i,v}}s_{i,v}\right)+n_{m,k},
\end{equation}
\begin{equation}
y_{e,j}=d^{-\alpha/2}_{e,j}\mathbf{g}^H_{e,j}\sum_{i=1}^{M}\left(\mathbf{w}_i\sum_{v=1}^{K_i}\sqrt{P\theta_{i,v}}s_{i,v}\right)+n_{e,j},
\end{equation}
where $n_{m,k}\sim \mathcal{CN}(0,\sigma^2_b)$ and $n_{e,j}\sim \mathcal{CN}(0,\sigma^2_e)$ are the receiver noises at user $k$ in cluster $m$ and Eve $j$, respectively.

To suppress the interference among clusters, we employ zero-forcing beamforming based on $\{\hat{\mathbf{g}}_{m}\}_{m=1}^M$, and the normalized beamformer $\mathbf{w}_m$ for cluster $m$ is chosen to satisfy~\cite{J_Kout12DownSDMA,J_Wang19SecureNOMI}
\begin{equation}\label{eq:zf_beamformer}
\hat{\mathbf{g}}^H_n\mathbf{w}_m=0, \quad \forall n\neq m, ~ n\in\{1,\ldots,M\}.
\end{equation}
On the other hand, to suppress intra-cluster interference, we can apply successive interference cancellation (SIC) in each cluster based on $\{d_{m,k}\}_{k=1}^{K_m}$~\cite{J_Yang16ImpefectCSI}.
Since the user locations are fixed, the distance and path loss are deterministic, and $\{d_{m,k}\}_{k=1}^{K_m}$ can be obtained through the Global Positioning System (GPS) or estimated based on signal propagation model~\cite{J_Andes95_Prog}. Without loss of generality, it is assumed that $d_{m,1}<\dots<d_{m,K_m}$.
Furthermore, since the interference cancellation is conducted at the user side for power domain downlink NOMA~\cite{J_Ding17App}, perfect SIC could be obtained at users~\cite{J_Yang16ImpefectCSI,J_Fang17Joint}.
As a result, the received signal-to-interference-plus-noise ratio (SINR) at user $k$ in cluster $m$ is given by~\cite{J_Chen18NOMASecure}
\begin{align}
&\rho_{m,k} \nonumber\\
=&\frac{|\frac{\mathbf{g}^H_{m,k}\mathbf{w}_m}{\mu_{m,k}}|^2 \theta_{m,k}}{|\frac{\mathbf{g}^H_{m,k} \mathbf{w}_m}{\mu_{m,k}}|^2\sum\limits_{i=1}^{k-1}\theta_{m,i}+{P_m}\sum\limits_{v\neq m}|\frac{\mathbf{g}^H_{m,k}\mathbf{w}_v}{\mu_{m,k}}|^2+\frac{1}{\gamma_{m,k}}}\label{eq:phi_User}\\
=&\frac{|\frac{\mathbf{g}^H_{m,k}\mathbf{w}_m}{\mu_{m,k}} |^2\theta_{m,k}}{|\frac{{\mathbf{g}}^H_{m,k}\mathbf{w}_m}{\mu_{m,k}} |^2\!\sum\limits_{i=1}^{k-1}\!\theta_{m,i}\!+\!P_m\|\!\frac{\sin \!\beta_{m,k}{\mathbf{g}}_{m,k}}{\mu_{m,k}}\!\|^2\!\sum\limits_{v\neq m} \!|\mathbf{e}^H_{m,k}\! \mathbf{w}_v|^2\!+\!\frac{1}{\gamma_{m,k}}},\label{eq:SIR_User}
\end{align}
where~\eqref{eq:SIR_User} is obtained by putting~\eqref{eq:rela_3_gk_channel_imperfect} and~\eqref{eq:zf_beamformer} into~\eqref{eq:phi_User}, $\gamma_{m,k}=P\mu^2_{m,k}d^{-\alpha}_{m,k}/\sigma^2_b$, $P_m=\sum_{k=1}^{K_m}\theta_{m,k}$ is the transmit power allocated to cluster $m$, and it satisfies $\sum_{m=1}^{M}P_m=1$.
In contrast, it is assumed that all Eves have no information about the decoding order in the SIC detector for a cluster. Accordingly, they cannot perform SIC within a cluster, and the corresponding SINR of eavesdropping user $k$ in cluster $m$ at Eve $j$ is given by~\cite{J_Chen18NOMASecure}
\begin{equation}\label{eq:SIR_Eve}
q^j_{m,k}=\frac{|\frac{{\mathbf{g}}^H_{e,j} \mathbf{w}_m}{\mu_{e,j}}|^2\theta_{m,k}}{|\frac{{\mathbf{g}}^H_{e,j}\mathbf{w}_m}{\mu_{e,j}}|^2\sum\limits_{i\neq k}\theta_{m,i}+{P_m}\sum\limits_{v\neq m}|\frac{{\mathbf{g}}^H_{e,j}\mathbf{w}_v}{\mu_{e,j}}|^2+\frac{1}{\gamma_{e,j}}},
\end{equation}
where $\gamma_{e,j}=P\mu^2_{e,j}d^{-\alpha}_{e,j}/\sigma^2_e$.
In light of~\eqref{eq:SIR_User} and~\eqref{eq:SIR_Eve}, the channel capacity for user $k$ in cluster $m$ and the corresponding eavesdropping capacity at Eve $j$ are given by $\log_2\left(1+\rho_{m,k}\right)$ and $\log_2\left(1+q^j_{m,k}\right)$, respectively.

For the main channel, the messages can be reliably received by user $k$ in cluster $m$ when the corresponding transmission rate $R_{m,k}$ satisfies $R_{m,k}\leq \log_2\left(1+\rho_{m,k}\right)$~\cite{J_Wynner75The}. However, since there is only one CDI feedback per cluster, the inter-cluster interference cannot be perfectly removed from user point-of-view. The residual interference is reflected by the term $\sum_{v\neq m} |\mathbf{e}^H_{m,k} \mathbf{w}_v|^2$ in~\eqref{eq:SIR_User}, and therefore $\log_2\left(1+\rho_{m,k}\right)$ is inaccurately known by the BS. As a result, reliable transmission cannot always be guaranteed because it is possible that $R_{m,k}>\log_2\left(1+\rho_{m,k}\right)$. Hence, the connection outage probability (COP) of user $k$ in cluster $m$ is expressed as~\cite{J_Hu18SecureIoT}
\begin{equation}\label{eq:ROP_user_k_ori}
    \begin{split}
        \mathrm{COP}: \quad p^{m,k}_{co}
        &=\mathrm{Pr}\left\{R_{m,k}>\log_2\left(1+\rho_{m,k}\right)\right\}.
    \end{split}
\end{equation}
For the $j^{th}$ Eve's channel, since the BS does not know its perfect CSI due to the passive nature of Eve $j$, the knowledge of $q^j_{m,k}$ is uncertain~\cite{J_Wang17SeMulti}.  Consequently, a secrecy outage event occurs at the BS when $\log_2\left(1+q^j_{m,k}\right)$ exceeds the redundancy rate of user $k$ in cluster $m$, denoted by $D^j_{m,k}$, and the secrecy outage probability (SOP) of user $k$ in cluster $m$ for Eve $j$ is given by~\cite{J_Wang17SeMulti}
\begin{equation}\label{eq:SOP_Cons}
\begin{split}
\mathrm{SOP}: \quad p^{m,k,j}_{so}
&=\mathrm{Pr}\left\{D^j_{m,k}<\log_2\left(1+q^j_{m,k}\right)\right\}.
\end{split}
\end{equation}
Considering the non-collaborative eavesdropping model, in which all Eves do not exchange their observations or outputs, the achievable secrecy rate for user $k$ in cluster $m$ is given by~\cite{J_multi08_eves} 
$\min\limits_{1\leq j\leq J}~[R_{m,k}-D^j_{m,k}]^+$,
which is the minimum over the secrecy rates achieved by all Eves.

For user $k$ in cluster $m$, maximizing $1-p^{m,k}_{co}$ would improve reliable transmission while maximizing the achievable secrecy rate $\min\limits_{1\leq j\leq J}~[R_{m,k}-D^j_{m,k}]^+$ would improve secure transmission. Considering both the reliability and security requirements for all users, we aim to maximize 
$\sum_{m=1}^{M}\sum_{k=1}^{K_m}\left(1-p^{m,k}_{co}\right)\min\limits_{1\leq j\leq J}\left[R_{m,k}-D^j_{m,k}\right]^+$, which is the security guaranteed sum-rate. Since it is known that $p^{m,k}_{co}$ in~\eqref{eq:ROP_user_k_ori} is independent of $D^j_{m,k}$, the security guaranteed sum-rate is equivalent to $\min\limits_{1\leq j\leq J}\{ \sum_{m=1}^{M}\sum_{k=1}^{K_m}\left(1-p^{m,k}_{co}\right)\! 
[R_{m,k}-D^j_{m,k}]^+ \}$. Considering the constraints of outage probability, the transmission design is thus given by the following optimization problem
\begin{subequations}\label{eq:max_EE}
\begin{align}
\mathcal{P}0:\max_{\mathcal{S}}& \min\limits_{1\leq j\leq J}\left\{\! \sum_{m=1}^{M}\sum_{k=1}^{K_m}\left(1-p^{m,k}_{co}\right) \left[R_{m,k}-D^j_{m,k}\right]^+\!\right\}, \label{obj:P0_perform} \\
    \mathrm{s.t.}\quad
    & p^{m,k}_{co} \leq \delta, ~\forall m,k,  \label{P0:cop_initial} \\
     & p^{m,k,j}_{so} \leq \varepsilon, ~ \forall m, k, j, \label{P0:SOP_initial}  \\
     & \sum_{k=1}^{K_m}\theta_{m,k}={P_m}, ~ \forall m,  \label{eq:simplex_theta_power}
\end{align}
\end{subequations}
where $\mathcal{S}=\{R_{m,k}\geq 0,D^j_{m,k}\geq 0,\theta_{m,k}\geq 0\}$, and
$\delta\in(0,1)$ and $\varepsilon\in(0,1)$ are the predefined upper bounds representing the maximum tolerable COP and SOP, respectively.\footnote{This assumption can be easily extended to the scenario with specific upper bounds requirements for each user and each eavesdropper due to the parallel structure of our proposed algorithm shown in Section~\ref{Sec:IV}.}

Problem $\mathcal{P}0$ provides an elegant formulation to measure the secure transmission performance via~\eqref{obj:P0_perform} and the outage uncertainties via~\eqref{P0:cop_initial} and~\eqref{P0:SOP_initial}. It is more general than the previous problem formulation in secure NOMA transmission under perfect CSI of the main channel (i.e., setting $\delta=0$ and $J=1$ in $\mathcal{P}0$ reduces to the formulation in~\cite{J_Wang19SecureNOMI}).  By solving $\mathcal{P}0$, we obtain not only a transmission scheme $\mathcal{S}$ that maximizes the security guaranteed sum-rate, but also the maximum value of the security guaranteed sum-rate by computing the objective function at the optimized solution.

However, the challenges in solving $\mathcal{P}0$ lie in the non-concavity of the objective function and the probabilistic constraints, which do not admit simple closed-form expressions. To circumvent the aforementioned challenges, in the next section, we derive a closed-form expression of COP constraint and a tight approximation to the SOP constraint.

\section{Handling the Probabilistic Constraints in $\mathcal{P}0$}\label{Sec:III}

We first handle the main channel outage probability. Putting~\eqref{eq:SIR_User} into~\eqref{eq:ROP_user_k_ori}, $p^{m,k}_{co}$ is expressed as~\eqref{eq:COP_cf_inter}, shown at the top of this page,
\begin{figure*}
	\normalsize
\begin{align}
p^{m,k}_{co}
&=\mathrm{Pr}\left\{R_{m,k}>\log_2\left(1+\rho_{m,k}\right)\right\} \nonumber\\
&=\mathrm{Pr}\left\{2^{R_{m,k}}-1\!>\!\frac{|\frac{\mathbf{g}^H_{m,k}\mathbf{w}_m}{\mu_{m,k}} |^2\theta_{m,k}}{|\frac{{\mathbf{g}}^H_{m,k}\mathbf{w}_m}{\mu_{m,k}} |^2\sum\limits_{i=1}^{k-1}\theta_{m,i}+{P_m}\|\frac{{\mathbf{g}}_{m,k}}{\mu_{m,k}}\|^2\sin^2 \beta_{m,k}\sum\limits_{v\neq m} |\mathbf{e}^H_{m,k} \mathbf{w}_v|^2
	+\frac{1}{\gamma_{m,k}}}\right\} \label{eq:COP_cf_inter} \\
&=1-\exp\left(\frac{1-2^{R_{m,k}}}{\left(\theta_{m,k}-\left(2^{R_{m,k}}-1\right)\sum\limits_{i=1}^{k-1}\theta_{m,i}\right)2\gamma_{m,k}}\right)
\left(1+\frac{2^{R_{m,k}}-1}{\theta_{m,k}-\left(2^{R_{m,k}}-1\right)\sum\limits_{i=1}^{k-1}\theta_{m,i}}\frac{{P_m}2^{-\frac{B}{N-1}}}{2}\right)^{1-M},\forall m, k  \label{eq:COP_cf_final},
\end{align}
	\hrulefill
\end{figure*}
where~\eqref{eq:COP_cf_final} is derived in Appendix~\ref{ref:der_cop_mk}.

On the other hand, putting~\eqref{eq:SIR_Eve} into~\eqref{eq:SOP_Cons}, $p^{m,k,j}_{so}$ is rewritten as
\begin{align}\label{eq:SOP_origi}
&p^{m,k,j}_{so} \nonumber \\
=&\mathrm{Pr}\left\{D^j_{m,k}<\log_2\left(1+q^j_{m,k}\right)\right\}\nonumber\\
=&\mathrm{Pr}\left\{\!2^{D^j_{m,k}}-1\!<\!\frac{|\!\frac{{\mathbf{g}}^H_{e,j}\mathbf{w}_m}{\mu_{e,j}} \!|^2\theta_{m,k}}{|\!\frac{{\mathbf{g}}^H_{e,j}\!\mathbf{w}_m}{\mu_{e,j}} \!|^2\!\sum\limits_{i\neq k}\theta_{m,i}\!+\!P_m\!\sum\limits_{v\neq m}\!|\!\frac{{\mathbf{g}}^H_{e,j}\!\mathbf{w}_v}{\mu_{e,j}}\!|^2\!+\!\frac{1}{\gamma_{e,j}}}\!\right\}\nonumber \\
=&\mathrm{Pr}\left\{2^{D^j_{m,k}}-1<{\frac{\mathbf{g}_{e,j}^H}{\mu_{e,j}}\mathbf{\Lambda}\frac{{\mathbf{g}}_{e,j}}{\mu_{e,j}}}\right\},~\forall m,k,j,
\end{align}
where $\mathbf{\Lambda}\in \mathbb{C}^{N\times N}$ is given by
\begin{align}\label{eq:lamda_quadraticForm}
\mathbf{\Lambda}= &\gamma_{e,j}\theta_{m,k}\mathbf{w}_m\mathbf{w}^H_m-\gamma_{e,j}\left(2^{D^j_{m,k}}-1\right)\mathbf{w}_m\mathbf{w}^H_m\sum_{i\neq k}\theta_{m,i}\nonumber \\
&-\gamma_{e,j}{P_m}\left(2^{D^j_{m,k}}-1\right)\sum_{v\neq m}\mathbf{w}_v\mathbf{w}^H_v.
\end{align}
It is observed that $p^{m,k,j}_{so}$ is expressed in the cumulative distribution function (CDF) of indefinite quadratic forms with ${\mathbf{g}}_{e,j}/{\mu_{e,j}}\sim \mathcal{CN}(\mathbf{0},\mathbf{I}_N)$. Hence, a closed-form expression of~\eqref{eq:SOP_origi} is given by~\cite[eq. 30]{J_IQF_GausianYI16}
\begin{equation}\label{eq:sop_inter_U-k}
p^{m,k,j}_{so}=\sum_{i=1}^{N_p}\prod_{v\neq i}^{N}\left(1-\frac{\lambda_{v}}{\lambda_{i}}\right)^{-1}\!\exp\left(-\frac{2^{D^j_{m,k}}-1}{\lambda_{i}}\right)\!,\!\forall m,k,j,
\end{equation}
where $\{\lambda_{i}\}_{i=1}^N$ are the eigenvalues of $\mathbf{\Lambda}$ in descending order, and $N_p$ denotes the number of positive and distinct eigenvalues among $\{\lambda_i\}_{i=1}^N$. By virtue of~\eqref{eq:COP_cf_final} and~\eqref{eq:sop_inter_U-k}, $\mathcal{P}0$ can be equivalently transformed into $\mathcal{P}1$, shown at the top of the next page.  
\begin{figure*}
	\normalsize
\begin{subequations}
	\begin{align}
	\mathcal{P}1: \max_{\mathcal{S}}& \min\limits_{1\leq j\leq J}\left\{ \sum_{m=1}^{M}\sum_{k=1}^{K_m}\! \frac{\exp\left(-\frac{2^{R_{m,k}}-1}{\left(\theta_{m,k}-\left(2^{R_{m,k}}-1\right)\sum\limits_{i=1}^{k-1}\theta_{m,i}\right)2\gamma_{m,k}}\right)[R_{m,k}-D^j_{m,k}]^+}{\left(1+\frac{2^{R_{m,k}}-1}{\theta_{m,k}-\left(2^{R_{m,k}}-1\right)\sum\limits_{i=1}^{k-1}\theta_{m,i}}\frac{{P_m}2^{-\frac{B}{N-1}}}{2}\right)^{M-1}}
	\right\},  \label{obj:Sum_rate_P1} \\
	\mathrm{s.t.}~ &
	\exp\left(\frac{1-2^{R_{m,k}}}{\left(\theta_{m,k}-\left(2^{R_{m,k}}-1\right)\sum\limits_{i=1}^{k-1}\theta_{m,i}\right)2\gamma_{m,k}}\right)\!
	\left(1+\frac{(2^{R_{m,k}}-1){P_m}2^{-\frac{B}{N-1}}}{\left(\theta_{m,k}-\left(2^{R_{m,k}}-1\right)\sum\limits_{i=1}^{k-1}\theta_{m,i}\right)2}\right)^{1-M} \!\geq\! 1- \delta,  \forall m,k, \label{eq:pro_cons1_K} \\
	&\sum_{i=1}^{N_p}\prod_{v\neq i}^{N}\left(1-\frac{\lambda_{v}}{\lambda_{i}}\right)^{-1}\exp\left(-\frac{2^{D^j_{m,k}}-1}{\lambda_{i}}\right) \leq \varepsilon, ~\forall m,k,j, \label{ineq:cons_SOP_CF} \\
	& \sum_{k=1}^{K_m}\theta_{m,k}={P_m}, ~\forall m
	\end{align}
\end{subequations}
	\hrulefill
\end{figure*}

Since $D^j_{m,k}$ is independent of $R_{m,k}$ and $\theta_{m,k}$, maximizing the objective function~\eqref{obj:Sum_rate_P1} is thereby equivalent to minimizing $D^j_{m,k}$. 
Notice that $D^j_{m,k}$ appears not only in the objective function~\eqref{obj:Sum_rate_P1}, but also through $\lambda_{i}$ in~\eqref{ineq:cons_SOP_CF}, making~\eqref{ineq:cons_SOP_CF} intractable. In order to proceed, we employ the following lemma, which provides a tighter constraint than~\eqref{ineq:cons_SOP_CF}.
\begin{lemma}\label{lem:Bertin-Type}
The following constraint is a tighter constraint than the SOP constraint of~\eqref{ineq:cons_SOP_CF}:
\begin{equation}\label{ineq:slack_SOP_trans}
\begin{split}
&2^{D^j_{m,k}}-1
\geq \frac{\theta_{m,k}}{\kappa_{m,k,j}+\sum\limits_{i\neq k}\theta_{m,i}}\geq 0, ~\forall m,k,j, 
\end{split}
\end{equation}
where $\kappa_{m,k,j}$ is given by
\begin{equation}\label{eq:para_kmj}
\begin{split}
&\kappa_{m,k,j}\\
=&\frac{\frac{1}{\gamma_{e,j}} 
	\!+\!{P_m}\mathrm{Tr}\left(\!\sum\limits_{v\neq m}\mathbf{w}_v\mathbf{w}^H_v\!\right)\!-\! {P_m}\sqrt{2\ln(\varepsilon_k^{-1})} \|\!\sum\limits_{v\neq m}\mathbf{w}_v\mathbf{w}^H_v\!\|_F }{\left(1+\ln(\varepsilon_k^{-1})+\sqrt{2\ln(\varepsilon_k^{-1})}\right)\mathrm{Tr}\left(\mathbf{w}_m\mathbf{w}^H_m\right)}
\end{split}
\end{equation}
with tunable parameter $\{\varepsilon_k\in(0,1]\}_{k=1}^{K_m}$, which could be different from $\varepsilon$ in~\eqref{ineq:cons_SOP_CF}. 
\end{lemma}
\begin{proof}
Please see Appendix~\ref{Appd-proof:Lema1}.
\end{proof}
\noindent
With~\eqref{ineq:cons_SOP_CF} replaced by~\eqref{ineq:slack_SOP_trans}, the minimum value of $D^j_{m,k}$ is obtained with equality of~\eqref{ineq:slack_SOP_trans} and is expressed as
\begin{equation}\label{eq:opt_r_mk}
\tilde{D}^j_{m,k}= \log_2\left(1+
\frac{\theta_{m,k}}{\kappa_{m,k,j} + \sum\limits_{i\neq k}\theta_{m,i}}\right),~\forall m,k,j.
\end{equation}
Since constraint~\eqref{ineq:slack_SOP_trans} is much tighter than~\eqref{ineq:cons_SOP_CF} in $\mathcal{P}1$, the obtained minimum $\tilde{D}^j_{m,k}$ is a conservative solution to $\mathcal{P}1$. 
To refine the solution to $\mathcal{P}1$, we can elaborately adjust $\varepsilon_k$ to reduce the conservatism and the discussion on how to select $\varepsilon_k$ is deferred to the end of Section~\ref{Sec:V}. Although replacing the original SOP with a tighter constraint would result in a solution $\tilde{D}^j_{m,k}$, it is still challenging to solve $\mathcal{P}1$, since $R_{m,k}$ and $\theta_{m,k}$ are strongly coupled in the objective function and the non-convex constraint~\eqref{eq:pro_cons1_K}.
To this end, we decouple the problem $\mathcal{P}1$ by introducing an auxiliary variable and provide a conventional solution in the next section.

\section{Problem Decoupling and Conventional Solution}\label{Sec:IV}
First, we introduce an auxiliary variable 
\begin{equation}\label{eq:aux_int_xi}
\xi_{m,k}=\frac{2^{R_{m,k}}-1}{\theta_{m,k}-\left(2^{R_{m,k}}-1\right)\sum\limits_{i=1}^{k-1}\theta_{m,i}}\geq 0,~\forall m,k.
\end{equation}
Putting $\xi_{m,k}$ and $\tilde{D}^j_{m,k}$ into~\eqref{obj:Sum_rate_P1}, the objective function in $\mathcal{P}1$ can be rewritten as 
\begin{equation}\label{eq:obj_trans_twoVars}
 \min\limits_{1\leq j\leq J}\left\{ \sum_{m=1}^{M}
\sum_{k=1}^{K_m} \frac{\left[\log_2\left(\frac{1+\frac{\xi_{m,k}\theta_{m,k}}{1+\xi_{m,k}\sum\limits_{i=1}^{k-1}\theta_{m,i}}}{1+\frac{\theta_{m,k}}{ \kappa_{m,k,j} + \sum\limits_{i\neq k}\theta_{m,i}} }\right)\right]^+}{\exp\left(\frac{\xi_{m,k}}{2\gamma_{m,k}}\right)\!\left(1+\frac{\xi_{m,k}P_m}{2^{\frac{B}{N-1}+1}}\right)^{M-1}}
\right\}.
\end{equation}
Furthermore, the constraint~\eqref{eq:pro_cons1_K} can be rewritten as
\begin{equation}\label{ineq:COP_aux_P1}
\exp\left(-\frac{\xi_{m,k}}{2\gamma_{m,k}}\right)\left(1\!+\!\frac{\xi_{m,k}P_m}{2^{\frac{B}{N-1}+1}}\right)^{1-M}\!
 \geq \! 1- \delta,  ~\forall m,k.
\end{equation}
Based on~\eqref{ineq:COP_aux_P1}, we can establish the feasible set of $\xi_{m,k}$ with the following lemma, which is proved in Appendix~\ref{proof:lem-cf_z_ub}.
\begin{lemma}\label{lem:clo-fom_zub}
    The feasible set of $\xi_{m,k}$ is $[0,\xi^{ub}_{m,k}]$, where $\xi^{ub}_{m,k}$ is given by
    \begin{equation}\label{eq:up_xi_CF}
    \begin{split}
    \xi^{ub}_{m,k}=&2\gamma_{m,k}(M-1)W_0\left(\frac{2^{\frac{B}{N-1}}\exp\left(\frac{2^{\frac{B}{N-1}}}{\gamma_{m,k}(M-1)P_m}\right)}{\gamma_{m,k}(M-1){P_m}(1-\delta)^{\frac{1}{M-1}}}\right)\\
    &-\frac{2^{\frac{B}{N-1}+1}}{{P_m}}.
    \end{split}
    \end{equation}
\end{lemma}

With~\eqref{eq:obj_trans_twoVars}-\eqref{eq:up_xi_CF}, $\mathcal{P}1$ (after~\eqref{ineq:cons_SOP_CF} tightened by~\eqref{ineq:slack_SOP_trans}) is  transformed into the following problem\footnote{Since $\mathcal{P}2$ is not equivalent to $ \mathcal{P}0$ due to a safe approximation, the final obtained solution to $\mathcal{P}2$ is a feasible solution to $\mathcal{P}0$.}
\begin{subequations}\label{opt:P2_initial}
    \begin{align}
    \mathcal{P}2: \!\max_{\{{\xi}_{m,k},\theta_{m,k}\geq 0\}_{k=1}^{K_m}} & 
   \! \min\limits_{1\leq j\leq J}\!\left\{ \sum\limits_{m=1}^{M}
   \sum_{k=1}^{K_m} \frac{\exp\left(-\frac{\xi_{m,k}}{2\gamma_{m,k}}\right)   }{\left(1+\frac{\xi_{m,k}P_m}{2^{\frac{B}{N-1}+1}}\right)^{M-1}}\right. \nonumber \\
  & \left.  \times
   \left[\log_2\left(\frac{1+\frac{\xi_{m,k}\theta_{m,k}}{1+\xi_{m,k}\sum\limits_{i=1}^{k-1}\theta_{m,i}}}{1+\frac{\theta_{m,k}}{ \kappa_{m,k,j} + \sum\limits_{i\neq k}\theta_{m,i}} }\right)\right]^+
   \right\}, \label{eq:obj_thr}\\
    \mathrm{s.t.}\quad
    &0\leq  \xi_{m,k} \leq \xi^{ub}_{m,k}, ~ \forall m,k, \label{ineq:ini_fea-xi}  \\
    & \sum_{k=1}^{K_m}\theta_{m,k}={P_m}, ~ \forall m. \label{eq:init-fea-the}
    \end{align}
\end{subequations}
From~\eqref{eq:para_kmj}, it is known that $\kappa_{m,k,j}$ is independent of $\xi_{m,k}$ and $\theta_{m,k}$. Hence, the operations of maximization and minimization in $ \mathcal{P}2$ can be interchanged. Furthermore, since $\{\kappa_{m,k,j}\}_{j=1}^J$ are independent of one another, we can solve $\mathcal{P}2$  by separately solving $J$ independent maximization problems and selecting the minimum value. Moreover, for subproblem $j$ (corresponding to Eve $j$), the objective function of $\mathcal{P}2$ contains a summation on $m$, and $\{{\xi}_{m,k},\theta_{m,k}\}_{m=1}^M$ are independent of one another. Hence, subproblem $j$ further reduces to $M$ parallel subproblems, where the subproblem of cluster $m$ with respect to Eve $j$ is expressed as $\mathcal{P}2^{[m,j]}$, shown at the top of the next page.  
\begin{figure*}
	\normalsize
\begin{subequations}\label{opt:sub_inti}
	\begin{align}
	\mathcal{P}2^{[m,j]}: \max_{\{{\Xi}^{(m,j)}_{k},\Theta^{(m,j)}_{k}\geq 0\}_{k=1}^{K_m}} & 
	\sum_{k=1}^{K_m} \frac{\exp\left(-\frac{\Xi^{(m,j)}_{k}}{2\gamma_{m,k}}\right)}{\left(1+\Xi^{(m,j)}_{k}\frac{{P_m}2^{-\frac{B}{N-1}}}{2}\right)^{M-1}}  
	\left[\log_2\left(\frac{1+\frac{\Xi^{(m,j)}_{k}\Theta^{(m,j)}_{k}}{1+\Xi^{(m,j)}_{k}\sum\limits_{i=1}^{k-1}\Theta^{(m,j)}_{i}}}{1+\frac{\Theta^{(m,j)}_{k}}{ \kappa_{m,k,j} + \sum\limits_{i\neq k}\Theta^{(m,j)}_{i}} }\right)\right]^+, \label{eq:max_thr}\\
	\mathrm{s.t.}\quad
	&0\leq  \Xi^{(m,j)}_{k} \leq \xi^{ub}_{m,k},~\forall k,  \label{ineq:cons_xi_orthant} \\
	& \sum_{k=1}^{K_m}\Theta^{(m,j)}_{k}={P_m}, \label{eq:simplex_theta}
	\end{align}
\end{subequations}
	\hrulefill
\end{figure*}
Due to the reverse of minimization and maximization in $ \mathcal{P}2$ and to emphasize that $m$ and $j$ are fixed in $\mathcal{P}2^{[m,j]}$,  $\xi_{m,k}$ and $\theta_{m,k}$ are respectively relabeled as ${\Xi}^{(m,j)}_{k}$ and $\Theta^{(m,j)}_{k}$. 

Furthermore, since the feasible set of $\mathcal{P}2^{[m,j]}$ is a Cartesian product of closed convex sets, the objective function and constraints in $\mathcal{P}2^{[m,j]}$ are decoupled when either $\{\Xi^{(m,j)}_{k}\}_{k=1}^{K_m}$ or
$\{\Theta^{(m,j)}_{k}\}_{k=1}^{K_m}$ is fixed. Accordingly, the optimization problem can be solved via the block coordinate ascent approach~\cite{C_Xie18BCA} for alternatively updating $\{\Xi^{(m,j)}_{k}\}_{k=1}^{K_m}$ and $\{\Theta^{(m,j)}_{k}\}_{k=1}^{K_m}$.

To be specific, when fixing $\{\Theta^{(m,j)}_{k}\}_{k=1}^{K_m}$, the subproblem of $\mathcal{P}2^{[m,j]}$  for updating $\{\Xi^{(m,j)}_{k}\}_{k=1}^{K_m}$ is formulated as the following multiple-ratio fractional programming (FP) problem~\cite{J_Tof17FP} 
\begin{subequations}
    \begin{align}
    \mathcal{Q}1: \quad & \max_{\{\Xi^{(m,j)}_{k}\}_{k=1}^{K_m}} \quad \sum_{k=1}^{K_m} \frac{A_k(\Xi^{(m,j)}_{k})}{B_k(\Xi^{(m,j)}_{k})}, \label{obj:FP_Q1}  \\
    \mathrm{s.t.}\quad
    &0 \leq  \Xi^{(m,j)}_{k} \leq \xi^{ub}_{m,k}, \quad \forall k, \label{ineq:set_xi_Q1}
    \end{align}
\end{subequations}
where $A_k(\Xi^{(m,j)}_{k})$ and $B_k(\Xi^{(m,j)}_{k})$ are given by
\begin{equation}
\begin{split}
A_k(\Xi^{(m,j)}_{k})=	\left[\log_2\left(\frac{1+\frac{\Xi^{(m,j)}_{k}\Theta^{(m,j)}_{k}}{1+\Xi^{(m,j)}_{k}\sum\limits_{i=1}^{k-1}\Theta^{(m,j)}_{i}}}{1+\frac{\Theta^{(m,j)}_{k}}{ \kappa_{m,k,j} + \sum\limits_{i\neq k}\Theta^{(m,j)}_{i}} }\right)\right]^+,
\end{split}
\end{equation}
\begin{equation}
B_k(\Xi^{(m,j)}_{k})=\exp\left(\frac{\Xi^{(m,j)}_{k}}{2\gamma_{m,k}}\right)
\left(1+\frac{\Xi^{(m,j)}_{k}P_m}{2^{\frac{B}{N-1}+1}}\right)^{M-1}.
\end{equation}
Since~\eqref{obj:FP_Q1} is a sum-of-ratios objective function, a number of approaches are able to handle this problem~\cite{J_Benson07FP,J_Bosh15SoR,J_Mult_FP01}. Among them, branch-and-bound algorithm is a popular approach that systematically subdivides the compact feasible interval of $\Xi^{(m,j)}_{k}$ to locate the global optimal solution of $\mathcal{Q}1$.
However, the price for this method is the high computational complexity when $K_m$ is large.

On the other hand, when fixing $\{\Xi^{(m,j)}_{k}\}_{k=1}^{K_m}$, the subproblem of $\mathcal{P}2^{[m,j]}$ for updating $\{\Theta^{(m,j)}_{k}\}_{k=1}^{K_m}$ is formulated as 
\begin{equation}\label{opt:int_update_theta}
\begin{split}
\mathcal{D}1: &\max_{\{\Theta^{(m,j)}_{k} \geq 0\}_{k=1}^{K_m}}
\sum_{k=1}^{K_m}\left[\!\log_2\left(\frac{1\!+\!\frac{\Xi^{(m,j)}_{k}\Theta^{(m,j)}_{k}}{1+\Xi^{(m,j)}_{k}\sum\limits_{i=1}^{k-1}\Theta^{(m,j)}_{i}}}{1\!+\!\frac{\Theta^{(m,j)}_{k}}{ \kappa_{m,k,j} + \sum\limits_{i\neq k}\Theta^{(m,j)}_{i}}}\right)\!\right]^+, \\
& \mathrm{s.t.}~~
\sum_{k=1}^{K_m}\Theta^{(m,j)}_{k}={P_m}.
\end{split}
\end{equation}
Notice that the objective function can be rewritten as
$[F_1(\{\Theta^{(m,j)}_{k}\}_{k=1}^{K_m})-F_2(\{\Theta^{(m,j)}_{k}\}_{k=1}^{K_m})]^+$,
where $F_1(\{\Theta^{(m,j)}_{k}\}_{k=1}^{K_m})$ and $F_2(\{\Theta^{(m,j)}_{k}\}_{k=1}^{K_m})$ are both concave functions and given by
\begin{equation}
\begin{split}
F_1(\{\Theta^{(m,j)}_{k}\}_{k=1}^{K_m})=&\sum_{k=1}^{K_m} \log_2\left(1+\Xi^{(m,j)}_{k}\sum_{i=1}^{k}\Theta^{(m,j)}_{i}\right)\\
&+\sum_{k=1}^{K_m}\log_2\left(\kappa_{m,k,j}+ \sum_{i\neq k}\Theta^{(m,j)}_{i}\right),
\end{split}
\end{equation}
\begin{align}
F_2(\{\Theta^{(m,j)}_{k}\}_{k=1}^{K_m})=&\sum_{k=1}^{K_m} \log_2\left(1+\Xi^{(m,j)}_{k}\sum_{i=1}^{k-1}\Theta^{(m,j)}_{i}\right)\nonumber \\
&+\sum_{k=1}^{K_m}\log_2\left(\kappa_{m,k,j}+{P_m}\right).
\end{align}
As a result, the objective function of $\mathcal{D}1$ can be expressed in a difference of convex (DC) form~\cite{J_An05Opt}, and $\mathcal{D}1$ is transformed into a DC programming problem. Then, by employing convex-concave procedure (CCP)~\cite{J_Lipp16CCP}, a suboptimal solution of $\mathcal{D}1$ can be obtained.
Since CCP needs to solve a convex quadratic problem by the interior-point method, its complexity order is $\mathcal{O}(K_m^3)$ in each iteration~\cite{Ben-TalA01}, leading to heavy-computation when $K_m$ is large.

\section{Massive Access with First-Order Algorithm}\label{Sec:V}
To overcome the high complexity challenge brought by the conventional method when the number of users $K_m$ in cluster $m$ is large, we introduce an efficient gradient-based algorithm in this section to significantly reduce the computational complexity. 
Specifically, we provide a local optimal point for updating $\{\Xi^{(m,j)}_{k}\}_{k=1}^{K_m}$ when $\{\Theta^{(m,j)}_{k}\}_{k=1}^{K_m}$ is fixed and a stationary point for updating $\{\Theta^{(m,j)}_{k}\}_{k=1}^{K_m}$ when $\{\Xi^{(m,j)}_{k}\}_{k=1}^{K_m}$ is fixed, providing an overall first-order algorithm under the framework of alternating maximization (AM)~\cite{J_Beck15AM}. 
The convergence property of the proposed algorithm is also proved in this section.

\subsection{Updating $\{\Xi^{(m,j)}_{k}\}_{k=1}^{K_m}$}
To handle the multiple-ratio FP problem $\mathcal{Q}1$, we first establish the following theorem. 
\begin{theorem}\label{tem:FP_CCP_def}
    $\mathcal{Q}1$ is a multiple-ratio concave-convex FP problem with concave functions $A_k(\Xi^{(m,j)}_{k})$, convex functions $B_k(\Xi^{(m,j)}_{k})$, and nonempty convex set. 
\end{theorem}
\begin{proof}
Please see Appendix~\ref{proof:them-FP_condi}.
\end{proof}
\noindent
Based on Theorem~\ref{tem:FP_CCP_def}, $\mathcal{Q}1$ can be equivalently transformed into~\cite{J_Yuwei18FP}
\begin{subequations}\label{opt:upda_xi_org}
    \begin{align}
    \mathcal{Q}2: \!\! \max_{\{\Xi^{(m,j)}_{k},y_k \in \mathbb{R}\}_{k=1}^{K_m}}  &\sum_{k=1}^{K_m}\!\underbrace{\left(\!2y_k\sqrt{ A_k(\Xi^{(m,j)}_{k})}-y_k^2B_k(\Xi^{(m,j)}_{k})\!\right)}_{:= g(\Xi^{(m,j)}_{k},y_k)}, \label{obj:q2_fun} \\
    \mathrm{s.t.}\quad
    &0 \leq  \Xi^{(m,j)}_{k} \leq \xi^{ub}_{m,k}, \quad \forall k, 
    \end{align}
\end{subequations}
where $\{y_k\}_{k=1}^{K_m}$ are auxiliary variables. When $\{\Xi^{(m,j)}_{k}\}_{k=1}^{K_m}$ are fixed, the optimal $\{y^{\dagger}_k\}_{k=1}^{K_m}$ in $\mathcal{Q}2$ are derived as
\begin{equation}\label{eq:update_aux_y}
y_k^{\dagger}= \sqrt{A_k(\Xi^{(m,j)}_{k})}\Big/B_k(\Xi^{(m,j)}_{k}),  \quad  k = 1,\ldots,K_m.
\end{equation}
On the other hand, when $\{y_k\}_{k=1}^{K_m}$ are fixed, due to the concavity of each $A_k(\Xi^{(m,j)}_{k})$ and convexity of each $B_k(\Xi^{(m,j)}_{k})$ from Theorem~\ref{tem:FP_CCP_def},  $g(\Xi^{(m,j)}_{k},y^{\dagger}_k)$ is concave in $\Xi^{(m,j)}_{k}$. As a result, $\mathcal{Q}2$ is a concave maximization problem over ${\Xi}^{(m,j)}_{k}$, and the optimal ${\Xi^{(m,j)}_{k}}^{\dagger}$ for maximizing $g(\Xi^{(m,j)}_{k},y^{\dagger}_k)$ is summarized in the following property, which is proved in Appendix~\ref{prrof:lema_opt_Q2}.
\begin{prop}\label{lem:sta_xi_mk_fp} 
    The optimal $\{{\Xi^{(m,j)}_{k}}^{\dagger}\}_{k=1}^{K_m}$ in $\mathcal{Q}2$ are given by
    \begin{equation}\label{eq:opt_Xi-Q2}
    {\Xi^{(m,j)}_{k}}^{\dagger} = \min\left\{ {\Xi^{(m,j)}_{k}}^{\diamond}, \xi^{ub}_{m,k}  \right\},\quad  k = 1,\ldots,K_m, 
    \end{equation}
    where ${\Xi^{(m,j)}_{k}}^{\diamond}$ satisfies~\eqref{eq:sta_xi_mk}, shown at the top of this page.
        \begin{figure*}
    	\normalsize
    	\begin{equation} \label{eq:sta_xi_mk}
\frac{\left(1+{\Xi^{(m,j)}_{k}}^{\diamond}\sum\limits_{i=1}^{k-1}\Theta^{(m,j)}_{i}\right)^{-1}
	\log^{-\frac{1}{2}}_2\left(\frac{\left(1+{\Xi^{(m,j)}_{k}}^{\diamond}\sum\limits_{i=1}^{k}\Theta^{(m,j)}_{i}\right)\left(\kappa_{m,k,j} + \sum\limits_{i\neq k}\Theta^{(m,j)}_{i} \right)}
	{\left(1+{\Xi^{(m,j)}_{k}}^{\diamond}\sum\limits_{i=1}^{k-1}\Theta^{(m,j)}_{i}\right)(\kappa_{m,k,j} +{P_m})}
	\right)
}{\left(1+{\Xi^{(m,j)}_{k}}^{\diamond}\sum\limits_{i=1}^{k}\Theta^{(m,j)}_{i}\right)\left(\frac{1}{2\gamma_{m,k}}+\frac{(M-1){P_m}2^{-\frac{B}{N-1}}}{2+{P_m}{\Xi^{(m,j)}_{k}}^{\diamond}2^{-\frac{B}{N-1}}}\right)B_k({\Xi^{(m,j)}_{k}}^{\diamond})\ln 2}=\frac{y_k^{\dagger}}{\Theta^{(m,j)}_{k}}
    	\end{equation}
    	\hrulefill
    \end{figure*}
\end{prop}
\noindent
Based on Proposition~\ref{lem:sta_xi_mk_fp}, ${\Xi^{(m,j)}_{k}}^{\dagger}$ can be efficiently found via bisection method. 

To sum up, the entire procedure for solving $\mathcal{Q}2$ is summarized in Algorithm~\ref{alg:CDM_iterative_RI}, which is essentially a cyclic block coordinate ascent method. Furthermore, since $g(\Xi^{(m,j)}_{k},y_k)$ is biconcave on $\Xi^{(m,j)}_{k}$ and $y_k$ with separable feasible sets, it converges to a local optimal point of $\mathcal{Q}2$~\cite{J_Wright15}. 
On the other hand, it is observed that variable ${\Xi^{(m,j)}_{k}}$ maximizes the objective function of $\mathcal{Q}1$ if and only if ${\Xi^{(m,j)}_{k}}^{\dagger}$ together with $y_k^{\dagger}$ maximizes the objective function of $\mathcal{Q}2$. Hence, the transformed $\mathcal{Q}2$ has equivalent objective value and solution with respect to the original problem $\mathcal{Q}1$. 
Based on the above discussion, we conclude the following property with respect to Algorithm~\ref{alg:CDM_iterative_RI}.
\begin{algorithm}[H]
    \caption{Iterative method for solving $\mathcal{Q}2$} \label{alg:CDM_iterative_RI}
    \begin{algorithmic}[1]
        \STATE Initialize $\{\Xi^{(m,j)}_{k}\}_{k=1}^{K_m}$ to a feasible value.
        \REPEAT
        \STATE Updating $\{y_k\}_{k=1}^{K_m}$ by~\eqref{eq:update_aux_y}. 
        \STATE Updating $\{\Xi^{(m,j)}_{k}\}_{k=1}^{K_m}$ by~\eqref{eq:opt_Xi-Q2}.
        \UNTIL Stopping criterion is satisfied. 
    \end{algorithmic}
\end{algorithm}
\begin{theorem}\label{lem:opt_xi_alg1}
    Algorithm~\ref{alg:CDM_iterative_RI} consists of a sequence of concave optimization problems, with the corresponding solutions converging to a local optimal point of $\mathcal{Q}1$.
\end{theorem}

\subsection{Updating $\{\Theta^{(m,j)}_{k}\}_{k=1}^{K_m}$}

Notice that the objective function of $\mathcal{D}1$ must be non-negative at optimality. Hence, the pointwise maximum $[\cdot]^+$ can be dropped, and $\mathcal{D}1$ is equivalent to 
\begin{equation}\label{D2: obj_theta}
\begin{split}
	\mathcal{D}2: &\!\max_{ \{\Theta^{(m,j)}_{k} \geq 0\}_{k=1}^{K_m}}
\sum_{k=1}^{K_m}\underbrace{\!
	\log_2\left(\frac{1+\frac{\Xi^{(m,j)}_{k}\Theta^{(m,j)}_{k}}{1+\Xi^{(m,j)}_{k}\sum\limits_{i=1}^{k-1}\Theta^{(m,j)}_{i}}}{1+\frac{\Theta^{(m,j)}_{k}}{ \kappa_{m,k,j} + \sum\limits_{i\neq k}\Theta^{(m,j)}_{i}}}\right)\!}_{:= f\left(\{\Theta^{(m,j)}_{k}\}_{k=1}^{K_m}\right)}, \\
&\mathrm{s.t.}~~
\sum_{k=1}^{K_m}\Theta^{(m,j)}_{k}={P_m}.
\end{split}
\end{equation}
Since $\mathcal{D}2$ has a continuously differentiable objective function and linear feasible set, it can be solved by the projected-gradient (PG) method~\cite{B_Bertseka97NP}, which alternatively performs an unconstrained gradient descent step and computes the projection of the unconstrained update onto the feasible set of the optimization problem. To be specific, the update of $\{\Theta^{(m,j)}_{k}\}_{k=1}^{K_m}$ at the $l^{th}$ iteration is given by
\begin{align}\label{eq:proj_grad}
&{\Theta^{(m,j)}_{k}}^{\diamond}\left(l+\frac{1}{2}\right)\nonumber\\
&={\Theta^{(m,j)}_{k}}^{\diamond}(l)+\mathcal{I} \sum_{i=1}^{K_m}\nabla_{\Theta^{(m,j)}_{i}} f, ~\forall k =1,\ldots,K_m,
\end{align}
where $\mathcal{I}$ is a constant step size chosen by Armijo rule to guarantee convergence~\cite[Prop. 2.3.3]{B_Bertseka97NP}, and $\nabla_{\Theta^{(m,j)}_{i}} f$ is the gradient of $f(\{\Theta^{(m,j)}_{k}\}_{k=1}^{K_m})$ at $\Theta^{(m,j)}_{i}, i\in \{1,\ldots,K_m\}$, with its explicit expression shown in~\eqref{eq:clo_explicit_dif} of Appendix~\ref{lem:CF_theta_AM}.
On the other hand, to project ${\Theta^{(m,j)}_{k}}^{\diamond}(l+1/2)$ onto the feasible set of $\mathcal{D}2$ to find its nearest feasible point ${\Theta^{(m,j)}_{k}}^{\diamond}(l+1)$, we have an equivalent optimization problem expressed as
\begin{align}\label{opt:proj_D2}
&{\Theta^{(m,j)}_{k}}^{\diamond}(l+1)\nonumber \\
&=\mathop{\arg\min}_{\{\Theta^{(m,j)}_{k}\}_{k=1}^{K_m}\in \mathcal{P}_{\mathcal{D}_2}}\! \sum_{k=1}^{K_m}\left\|\Theta^{(m,j)}_{k}\!-\!{\Theta^{(m,j)}_{k}}^{\diamond}\!\left(l+\frac{1}{2}\right)\right\|^2,
\end{align}
where $\mathcal{P}_{\mathcal{D}_2}=\{\{\Theta^{(m,j)}_{k}\}_{k=1}^{K_m}| \sum_{k=1}^{K_m}\Theta^{(m,j)}_{k}={P_m},{\Theta}^{(m,j)}_{k}\geq 0\}$ is the domain of $\mathcal{D}2$.
Since~\eqref{opt:proj_D2} is a convex optimization problem, a closed-form solution  can be derived based on Karush-Kuhn-Tucker (KKT) condition and is given by the following property, which is proved in Appendix~\ref{proof:lemma_proj_gra_theta}.
\begin{prop}\label{lem:Proj_gra_theta_k}
The optimal solution of~\eqref{opt:proj_D2} is given by
\begin{equation}\label{eq:pro_gra_optTheta}
\begin{aligned}
&{\Theta^{(m,j)}_{k}}^{\diamond}(l+1)\!=\!\Bigg[{\Theta^{(m,j)}_{k}}^{\diamond}\left(l+\frac{1}{2}\right)\\
&-\frac{1}{K_m}\left(\sum_{k=1}^{K_m}{\Theta^{(m,j)}_{k}}^{\diamond}\!\left(l+\frac{1}{2}\right)-{P_m}\right)\Bigg]^+\!,\forall k =1,\ldots,K_m.
\end{aligned}
\end{equation}
\end{prop}
\noindent
Based on~\eqref{eq:proj_grad} and Proposition~\ref{lem:Proj_gra_theta_k}, we can iteratively update $\{\Theta^{(m,j)}_{k}\}_{k=1}^{K_m}$, where the convergent point is guaranteed to be a stationary point of $\mathcal{D}2$~\cite[Prop. 2.3]{B_Bertseka97NP}. 
To sum up, the above PG method for solving $\mathcal{D}2$ is summarized in Algorithm~\ref{alg:dual_update_theta}.
\begin{algorithm}[H]
    \caption{PG method for solving $\mathcal{D}2$} 
    \begin{algorithmic}[1]\label{alg:dual_update_theta}
        \STATE Initialize with a feasible point $\{{\Theta^{(m,j)}_{k}}^{\diamond}(0)\}_{k=1}^{K_m}$ and set $l:=0$.
        \REPEAT
        \STATE Update $\{\Theta^{(m,j)}_{k}\}_{k=1}^{K_m}$ with the gradient iterate~\eqref{eq:proj_grad} and projection iterate~\eqref{eq:pro_gra_optTheta}.
        \STATE Update iteration: $l:=l+1$.
        \UNTIL Stopping criterion is satisfied.
    \end{algorithmic}
\end{algorithm}

\subsection{Tightness Refinement and Overall Algorithm}
With the update of $\{{\Xi}^{(m,j)}_{k}\}_{k=1}^{K_m}$ and $\{\Theta^{(m,j)}_{k}\}_{k=1}^{K_m}$ given by Algorithm~\ref{alg:CDM_iterative_RI} and~\ref{alg:dual_update_theta} respectively, subproblem $ \mathcal{P}2^{[m,j]}$ can be solved under the AM framework, with the convergence revealed by the following Theorem (proved in Appendix~\ref{proof_tem_convergence}).
\begin{theorem}\label{tem:conver_AM}
For a given $\{\varepsilon_k\in(0,1]\}_{k=1}^{K_m}$, and starting from a feasible solution of $ \mathcal{P}2^{[m,j]}$, the sequence of solutions generated by
alternatively executing Algorithms~\ref{alg:CDM_iterative_RI} and~\ref{alg:dual_update_theta} converges to a stationary point of problem $ \mathcal{P}2^{[m,j]}$.
\end{theorem}

However, as pointed out by Theorem~\ref{tem:conver_AM}, the solution of $ \mathcal{P}2^{[m,j]}$ would depend on the parameter $\varepsilon_k$ (via $\kappa_{m,k,j}$), which controls the tightness of the approximation in Lemma~\ref{lem:Bertin-Type}. In particular, if we solve  $\mathcal{P}2^{[m,j]}$ with the tunable parameter $\varepsilon_k$ varying from $0$ to $1$ (other simulation setting detailed in Section VI) and calculate $p^{m,k,j}_{so}$ by putting the solution of $ \mathcal{P}2^{[m,j]}$ into~\eqref{eq:sop_inter_U-k}, we could obtain $p^{m,k,j}_{so}$ as a function of $\varepsilon_k$, denoted by $p^{m,k,j}_{so}(\varepsilon_k)$, and the results are shown in Fig.~\ref{fig:refine_Solution} for a few selected $\{m,k,j\}$.
It is observed that $p^{m,k,j}_{so}$ after a safe approximation (named as Approximation) is usually far less than the tunable parameter $\varepsilon_k$, as shown in the gap between the diagonal black dotted line and colored lines.
However, since $p^{m,k,j}_{so}(\varepsilon_k)$ is nondecreasing in $\varepsilon_k$ as shown in Fig.~\ref{fig:refine_Solution}, we could use the bisection method~\cite{C_BoShen08Proba} to relieve the performance loss and locate a proper $\varepsilon_k\in [\varepsilon,1]$ such that $p^{m,k,j}_{so}(\varepsilon_k)$ is close to the required $\varepsilon$. 
\begin{figure}
	\centering
	\includegraphics[width=3.7in]{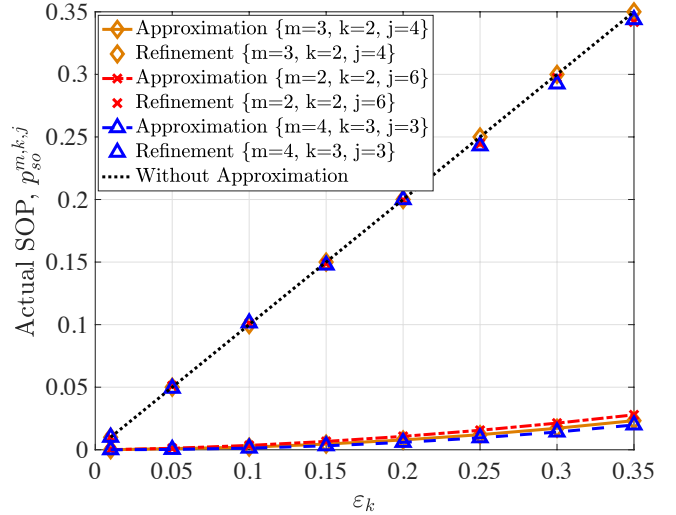}
	\caption{$p^{m,k,j}_{so}$ versus $\varepsilon_k$ with the basic simulation setting detailed in Section~\ref{Sec:VI}.}\label{fig:refine_Solution}
\end{figure}

Based on the above discussion, the proposed first-order algorithm for solving $ \mathcal{P}2^{[m,j]}$ with tightness parameter refinement is summarized in Algorithm~\ref{refinement method to P2} and the refined results (named as Refinement) are shown in Fig.~\ref{fig:refine_Solution}. It is observed that $p^{m,k,j}_{so}$ after refinement is very close to $\varepsilon_k$ without approximation. This indicates that although the solution set after approximation in Lemma~\ref{lem:Bertin-Type} is smaller, the corresponding solution set after refinement approximates the original one very well. 

Notice that Algorithm~\ref{refinement method to P2} consists of an outer bisection iteration and an inner AM iteration.
Since the outer bisection iteration must converge, together with Theorem~\ref{tem:conver_AM}, the overall Algorithm~\ref{refinement method to P2} is guaranteed to converge. For the computational complexity of Algorithm~\ref{refinement method to P2}, it is dominated by the inner AM iteration (i.e., from step 6 to step 9 in Algorithm~\ref{refinement method to P2}). To be specific, the computational complexity of Algorithm~\ref{alg:CDM_iterative_RI} is dominated by step 4 with the bisection search to update ${\Xi^{(m,j)}_{k}}$ at each iteration. Hence, the complexity order for updating $\{{\Xi}^{(m,j)}_{k}\}_{k=1}^{K_m}$ is $\mathcal{O}\left(K_m\ln(1/\varsigma)\right)$~\cite{Ben-TalA01}, where $\varsigma>0$ denotes the predefined searching resolution of the bisection method. On the other hand, the computational complexity of Algorithm~\ref{alg:dual_update_theta} is dominated by step 3 with the gradient iteration, which only involves the first-order differentiation. Therefore, the complexity order for updating $\{\Theta^{(m,j)}_{k}\}_{k=1}^{K_m}$ is $\mathcal{O}\left(K_m/\tau \right)$ with an accuracy of $\tau $~\cite{B_Bertseka97NP}. Since the complexity of Algorithm~\ref{refinement method to P2} is linear in $K_m$, it is suitable for massive access.

Noticing that $\mathcal{P}2$ consists of $MJ$ parallel subproblems in the form of $ \mathcal{P}2^{[m,j]}$, the overall algorithm for solving $\mathcal{P}2$ can be implemented in a parallel manner and is summarized in Algorithm~\ref{alg:overall_am_P3}, where the modern multi-core computing architecture can be leveraged for speeding up the computation. 
  
  \begin{algorithm}[H]
	\caption{Solution of $\mathcal{P}2^{[m,j]}$ with optimized tunable parameter $\{\varepsilon_k\}_{k=1}^{K_m}$} 
	\begin{algorithmic}[1]\label{refinement method to P2}
		\STATE \textbf{input}: the required $\varepsilon$ and a predefined searching resolution $z$.
		\STATE Initialize $\varepsilon^k_{\min}=\varepsilon$, $\varepsilon^k_{\max}=1$.
		\REPEAT
		\STATE Update tunable parameter $\varepsilon_k = (\varepsilon^k_{\min}+\varepsilon^k_{\max})/2 $.
		\STATE Initialize the feasible point of $\{{\Xi}^{(m,j)}_{k},\Theta^{(m,j)}_{k}\}_{k=1}^{K_m}$ based on~\eqref{ineq:cons_xi_orthant} and~\eqref{eq:simplex_theta}.
		\REPEAT
		\STATE Update $\{{\Xi}^{(m,j)}_{k}\}_{k=1}^{K_m}$ by using Algorithm~\ref{alg:CDM_iterative_RI}.
		\STATE Update $\{\Theta^{(m,j)}_{k}\}_{k=1}^{K_m}$ by using  Algorithm~\ref{alg:dual_update_theta}.
		\UNTIL Stopping criterion is satisfied.
		\STATE Calculate $p^{m,k,j}_{so}(\varepsilon_k)$ by putting the solution of $\mathcal{P}2^{[m,j]}$ into~\eqref{eq:sop_inter_U-k} for all $k$.
		\STATE\textbf{if} ~{$p^{m,k,j}_{so}(\varepsilon_k)<\varepsilon$}, \textbf{then} ~{$\varepsilon^k_{\min}=\varepsilon_k$},
		\STATE\textbf{else} ~{$\varepsilon^k_{\max}=\varepsilon_k$}.
		\STATE\textbf{end if} 
		\UNTIL $|p^{m,k,j}_{so}(\varepsilon_k)-\varepsilon|\leq z$.
		\STATE\textbf{output}: The maximizer $\{{{\Xi}^{(m,j)}_{k}}^*,{\Theta^{(m,j)}_{k}}^*\}_{k=1}^{K_m}$ with optimized $\{\varepsilon_k\}_{k=1}^{K_m}$.
	\end{algorithmic}
\end{algorithm}

\begin{algorithm}[H]
    \caption{The overall algorithm for solving $\mathcal{P}2$ } 
    \begin{algorithmic}[1]\label{alg:overall_am_P3}
        \STATE Solve $\mathcal{P}2^{[m,j]}$ in a parallel manner for all $m,j$ by using Algorithm~\ref{refinement method to P2}.
        \STATE Put $\{{{\Xi}^{(m,j)}_{k}}^*,{\Theta^{(m,j)}_{k}}^*\}_{k=1}^{K_m}$ for all $m,j$ into the objective function of~\eqref{eq:obj_thr}.
        \STATE Select $\hat{j}\in\{1,\ldots,J\}$ such that the objective function value of $\mathcal{P}2$ is the minimum.
        \STATE The maximizer $\xi^*_{m,k}={{\Xi}^{(m,\hat{j})}_{k}}^*$ and $\theta^*_{m,k}={\Theta^{(m,\hat{j})}_{k}}^*$ for all $m,k$.
    \end{algorithmic}
\end{algorithm}

\section{Numerical Results and Discussions}\label{Sec:VI}
In this section, we evaluate the secure transmission performance of the proposed algorithm through simulations. All simulations are performed on MATLAB R2017a on a Windows x64 desktop with 3.2 GHz CPU and 16 GB RAM. Each point in the figures is obtained by averaging over 100 simulation trials. Unless otherwise specified, the simulation set-up is as follows and kept throughout this section. 
We adopt the carrier frequency of 915 MHz and a carrier spacing of 200 kHz according to the 3GPP specification~\cite{LTE_carrierFQ13}. The path loss exponent is $\alpha = 2.5$ in the free space environment~\cite{B_Rappaport}. There are 100 users and 10 Eves in the whole system. All users are randomly distributed between 1 m and 100 m, while the location for Eve $j$ is fixed at $d_{e,j}=10/j$ m.
Once the large-scale fading parameters are generated, they are assumed to be known and fixed throughout the simulations. The small-scale fading vectors of all users and Eves are independently generated according to $ \mathcal{CN}(\mathbf{0},\mathbf{I}_N)$, i.e., $\mu_{m,k}=\mu_{e,j}=1$. 
The power allocation to cluster $m$ is set to $P_m=1/M$. The noise power at each user is set to $\sigma^2_b=0$ dB, and the noise power at each eavesdropper is set to $\sigma^2_e=5$ dB. To avoid repeating figure descriptions, the settings for $(\delta,\varepsilon,M,N,J,P)$ are provided in the caption of each figure.

\subsection{Performance of the Proposed First-Order Algorithm}
Firstly, we demonstrate the convergence of Algorithm~\ref{refinement method to P2} for solving $\mathcal{P}2^{[m,j]}$ with fixed $\{\varepsilon_k=0.1\}_{k=1}^{K_m}$. Since Algorithm~\ref{refinement method to P2} consists of multiple layers of iterations, stopping criterion for each layer depends on the relative change of the two consecutive objective function values (e.g., less than $10^{-4}$). The convergence of the inner loop of Algorithm~\ref{alg:CDM_iterative_RI} in terms of updating $\{\Xi^{(m,j)}_{k}\}_{k=1}^{K_m}$ is shown in Fig.~\ref{fig:Conver_alg1}. It is observed that the inner loop converges within 10 iterations under different numbers of $P$, which corroborates the results in Theorem~\ref{lem:opt_xi_alg1}.
When fixing the number of cluster $M$, the power allocated to each cluster increases with increasing $P$. As a result, the objective function value of $g(\Xi^{(m,j)}_{k},y_k)$ in $\mathcal{Q}2$ increases as $P$ increases. On the other hand, the convergence property of Algorithm~\ref{alg:dual_update_theta} is shown in Fig.~\ref{fig:Conver_alg2}. It is shown that the PG method converges after 300 iterations under different values of $P$. To verify the convergence of Algorithm~\ref{refinement method to P2}, Fig.~\ref{fig:Conver_alg3} shows the objective function value of $\mathcal{P}2^{[m,j]}$ versus the inner AM iteration. It is observed that AM converges rapidly within 25 iterations under different values of $P$, which corroborates the convergence result of Theorem~\ref{tem:conver_AM}.
\begin{figure*} 
	\centering
	\subfigure[]{ 
		\label{fig:Conver_alg1}
		\includegraphics[width=2.2in]{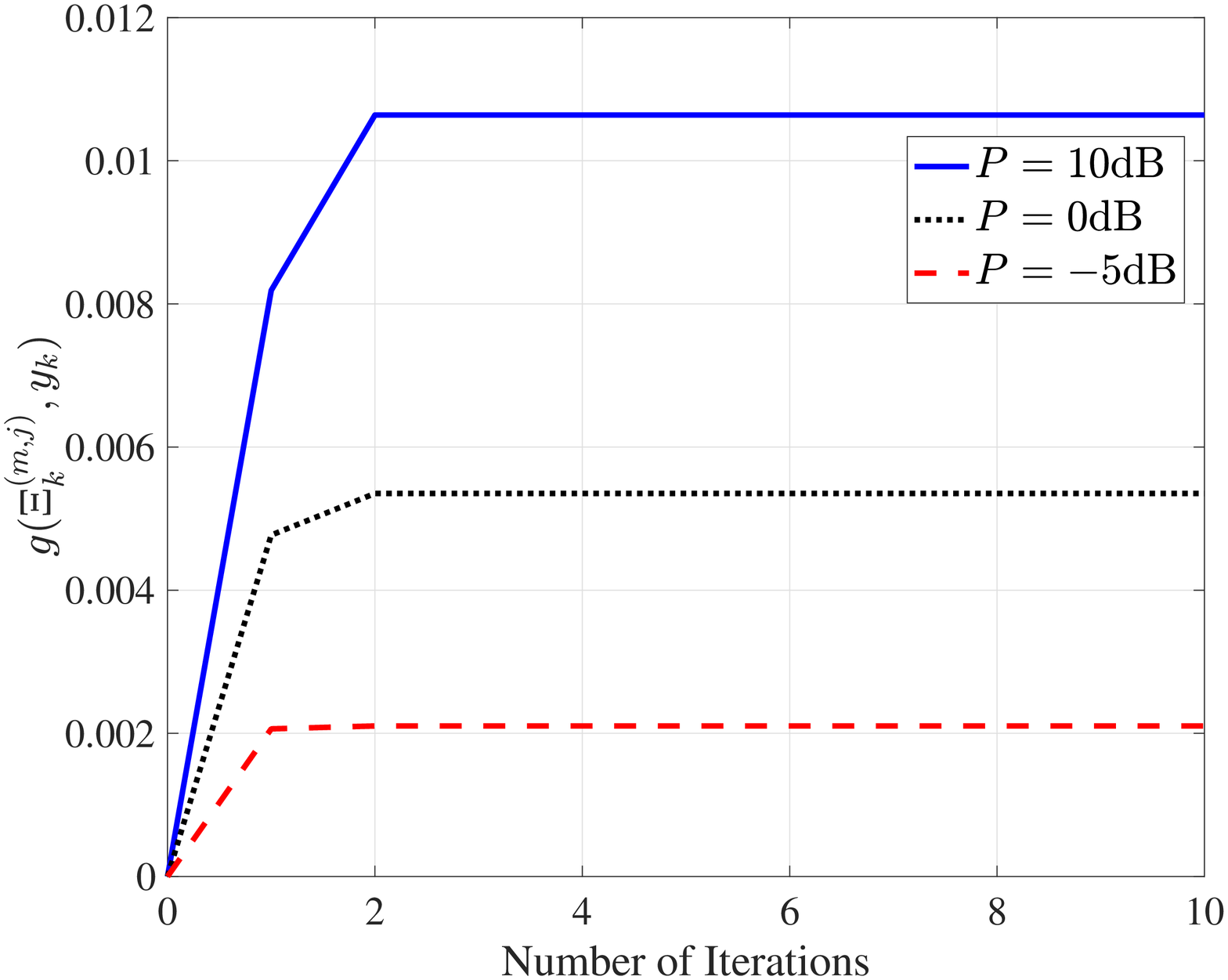}} \hspace{0in} 
	\subfigure[]{
		\label{fig:Conver_alg2} 
		\includegraphics[width=2.2in]{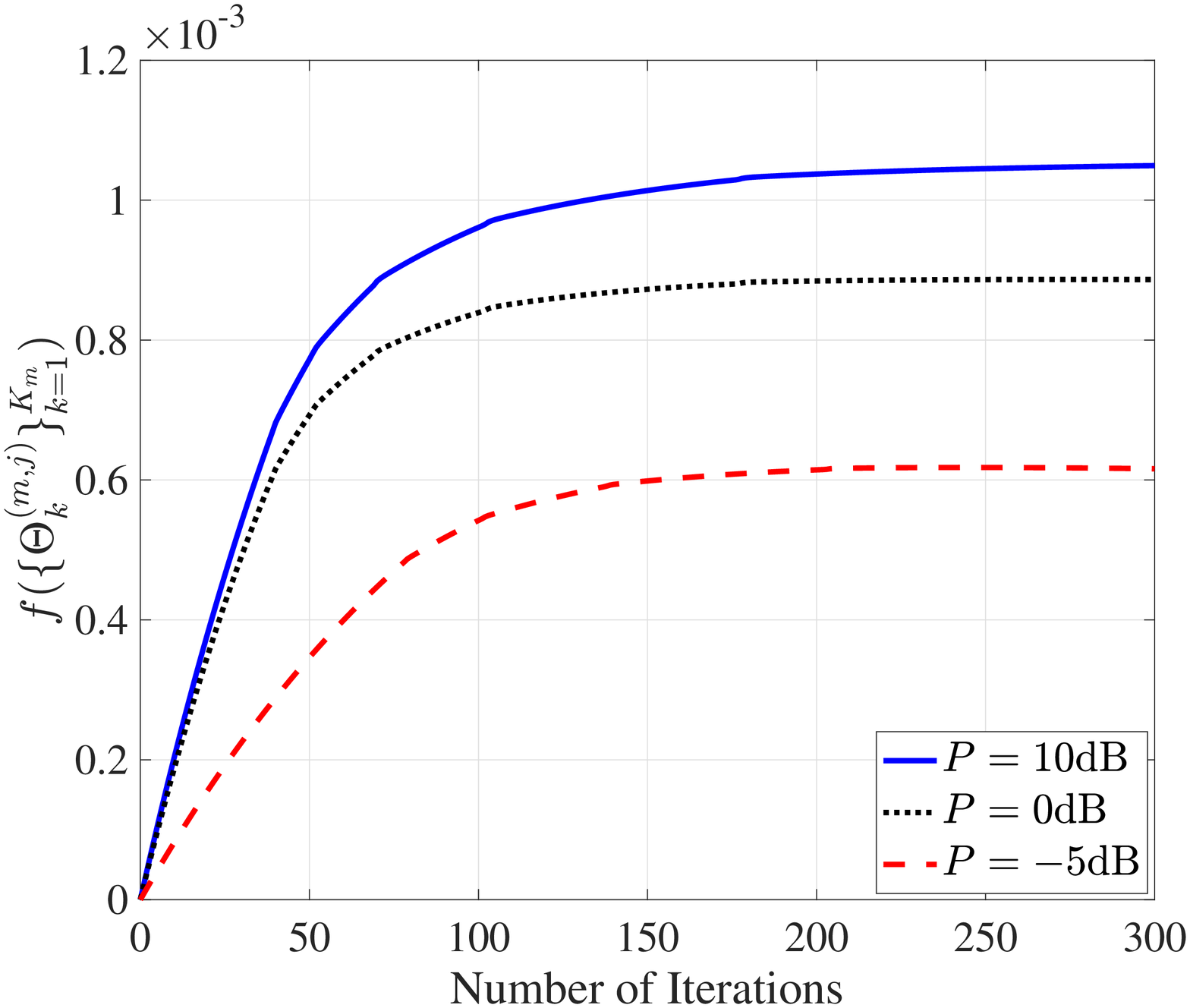}}  \hspace{0in} 
	\subfigure[]{
		\label{fig:Conver_alg3} 
		\includegraphics[width=2.2in]{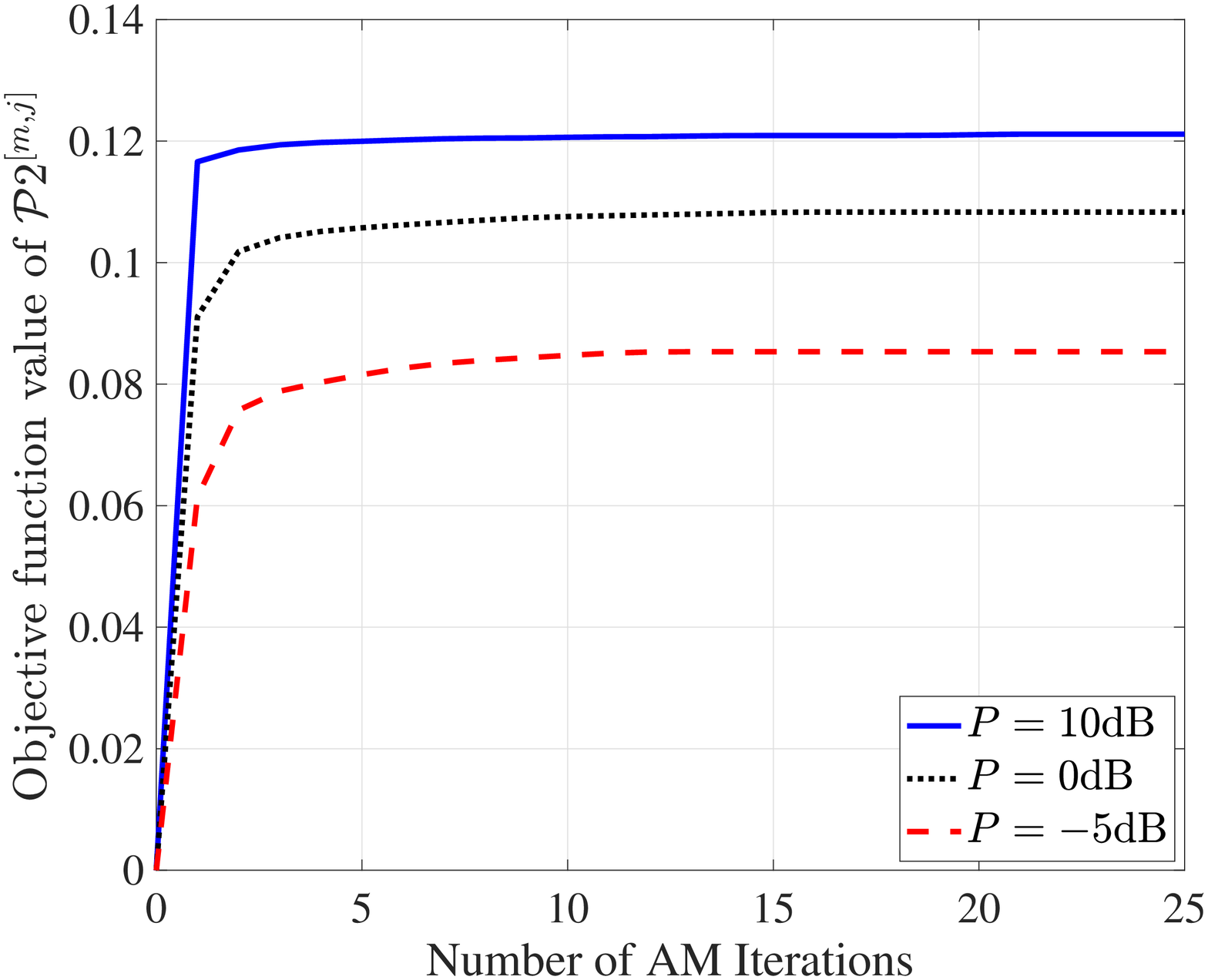}} 
	\caption{For a given $\{\varepsilon_k=0.1\}_{k=1}^{K_m}$ with $M=8$, $N=100$, $J=5$, $\delta=0.5$. (a)~The iterations of Algorithm 1.~(b)~The  iterations of Algorithm 2.~(c)~The AM iterations in Algorithm 3.}  
\end{figure*}

Next, to show the computational complexity advantage of the proposed algorithm for solving $\mathcal{P}2$, we compare Algorithm~\ref{alg:overall_am_P3} with the conventional method, which is the combination of the branch-and-bound algorithm and CCP. 
To be specific, $\{\Xi^{(m,j)}_{k}\}_{k=1}^{K_m}$ is updated based on branch-and-bound algorithm and the initial value is chosen according to box constraint~\eqref{ineq:set_xi_Q1}; $\{\Theta^{(m,j)}_{k}\}_{k=1}^{K_m}$ is updated based on CCP with the interior-point method and the initial value is chosen as $\{\Theta^{(m,j)}_{k}=P_m/K_m\}_{k=1}^{K_m}$. Besides, the convergence tolerance and maximum number of iterations for the conventional method are set to $10^{-4}$ and 100, respectively. As shown in Fig.~\ref{fig:time_compare}, compared with the conventional method, the proposed Algorithm~\ref{alg:overall_am_P3} reduces the computation time by at least two orders of magnitude.\footnote{The proposed algorithm has the potential of leveraging the modern multi-core computing architecture, and more efficient programming languages (C or Assembler) to further speed up the computation in practical implementation.}  
On the other hand, Fig.~\ref{fig:Perf_alg_compare} shows that the proposed Algorithm~\ref{alg:overall_am_P3} achieves almost the same security guaranteed sum-rate as the conventional method under different values of $P$. 
Due to the complexity advantage, we only provide the solution to $\mathcal{P}2$ obtained via Algorithm~\ref{alg:overall_am_P3} in the following discussion.

\subsection{Performance Comparisons with Other Multiplexing Schemes}

To show the performance advantages of the proposed scheme by employing the power domain multiplexing, we make a comparison with the orthogonal multiple access scheme~\cite{J_19XuTDMACompare}, where we employ zero-forcing beamforming among clusters and time division multiple access (TDMA) within the cluster. To begin with, we illustrate the impact of COP constraint on system performance, where all security guaranteed sum-rates are increasing in $\delta$ as shown in Fig.~\ref{fig:Thrpt_delta}. Moreover, the sum-rate advantage of the proposed scheme with respect to that of TDMA becomes more prominent as $\delta$ increases and then maintains at a large margin.
A heuristic explanation of this phenomenon is that $\Xi^{(m,j)}_{k}$ is no longer constrained by its upper bound $\xi^{ub}_{m,k}$ when $\delta$ is large according to~\eqref{eq:up_xi_CF}. 
When comparing two schemes in terms of the number of clusters, Fig.~\ref{fig:Thr_feedbackBits} shows that the proposed scheme significantly improves the security guaranteed sum-rate of the TDMA scheme under different values of $M$. When fixing $P$, the power allocated to each cluster decreases as $M$ increases. As a result, the sum-rate decreases as $M$ increases. 
Furthermore, with an increase of $J$, the performance degrades due to more Eves in the system.

To show the importance of considering imperfect CSI, we compare the proposed scheme with NOMA that does not consider the imperfect CSI~\cite{J_Wang19SecureNOMI}. 
To make a fair comparison with~\cite{J_Wang19SecureNOMI}, we simulate both schemes under the same security requirement and select $J=1$.
The security guaranteed sum-rate versus $\varepsilon$ and $P$ are provided in Fig.~\ref{fig:Thr_epsilon} and Fig.~\ref{fig:Througput_SNR}, respectively. It can be seen that the proposed scheme always achieves significantly higher sum-rates than NOMA ignoring CSI uncertainty.
Finally, Fig.~\ref{fig:Througput_SNR} shows that the security guaranteed sum-rate decreases as $N$ increases, which might seem counterintuitive. However, this is due to the coarser CDI when $N$ increases under a fixed $B$. This phenomenon can be seen in~\eqref{eq:obj_trans_twoVars}, where $B$ and $N$ appear in a ratio. This also suggests using more feedback bits would remedy the performance loss.

\begin{figure*} 
    \centering
    \subfigure[]{ 
        \label{fig:time_compare}%% label for first subfigure
        \includegraphics[width=3in]{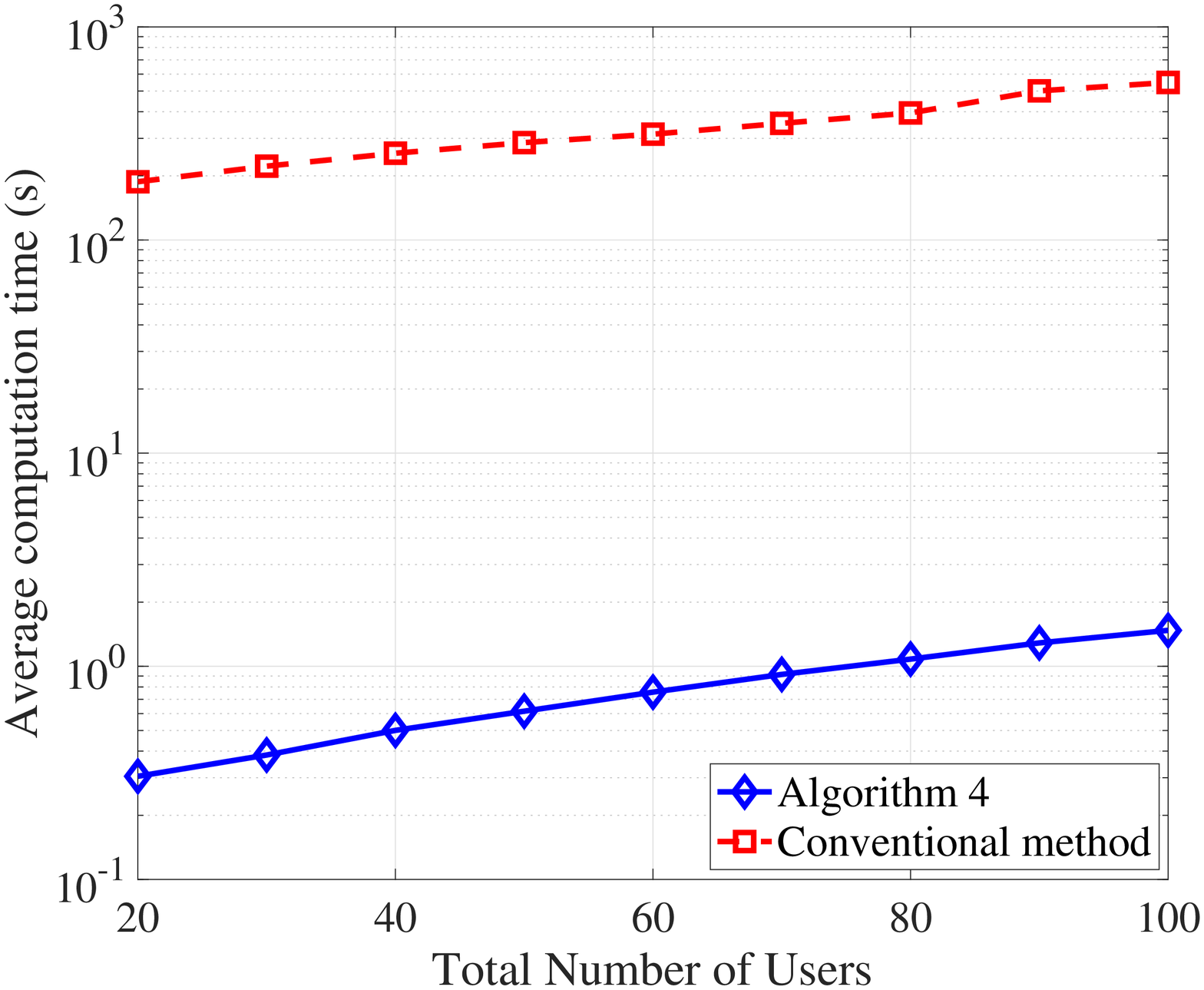}} \hspace{0.0in} 
    \subfigure[]{
        \label{fig:Perf_alg_compare} %% label for second subfigure
        \includegraphics[width=3in]{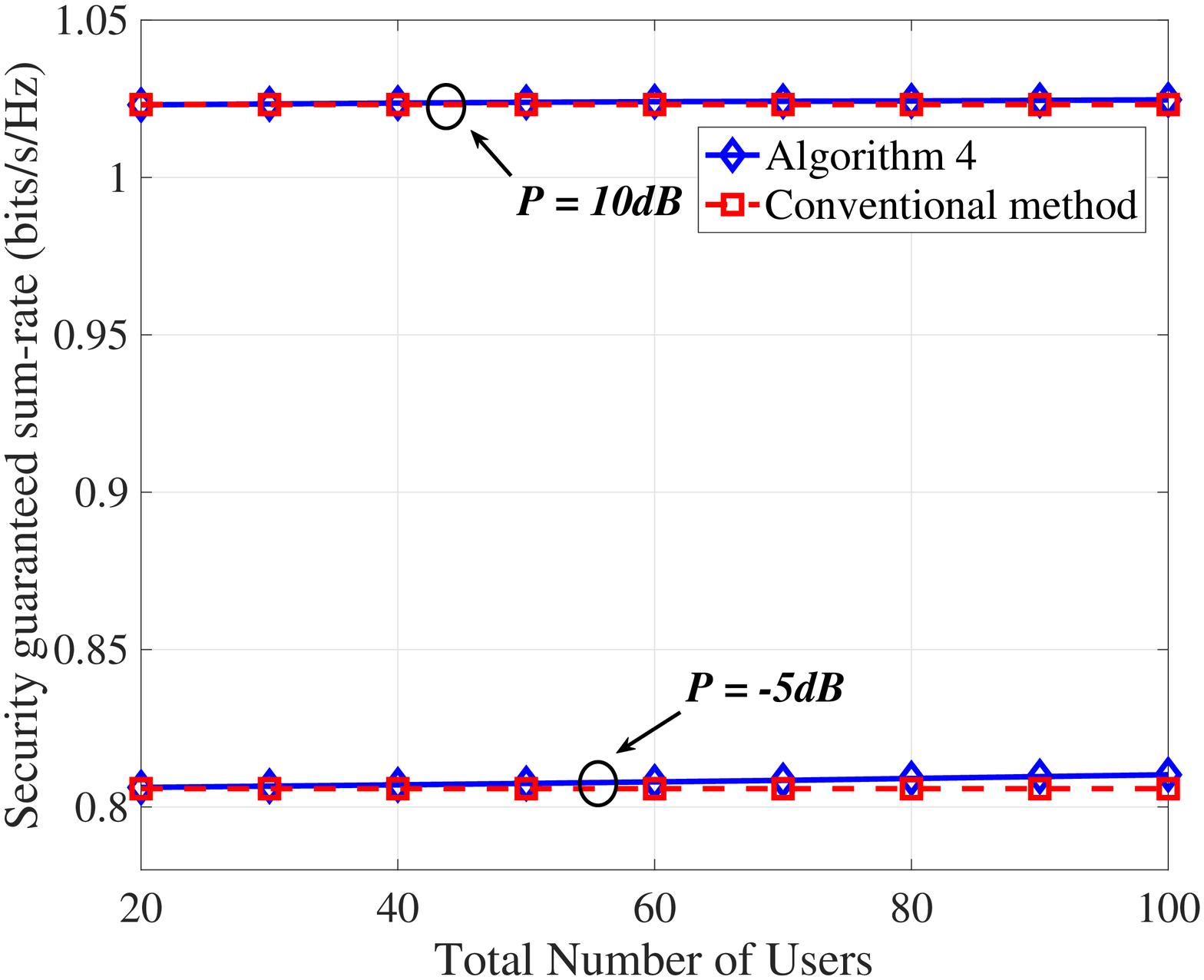}} 
    \caption{Performance comparison with the conventional method with $M=8$, $N=100$, $J=5$, $P=10\mathrm{dB}$, $\varepsilon=0.1$, $\delta=0.5$. (a)~Average computation time versus total number of users.~(b)~Security guaranteed sum-rate versus total number of users.}  
\end{figure*}

\begin{figure*} 
    \centering
    \subfigure[]{ 
        \label{fig:Thrpt_delta}%% label for first subfigure
        \includegraphics[width=3in]{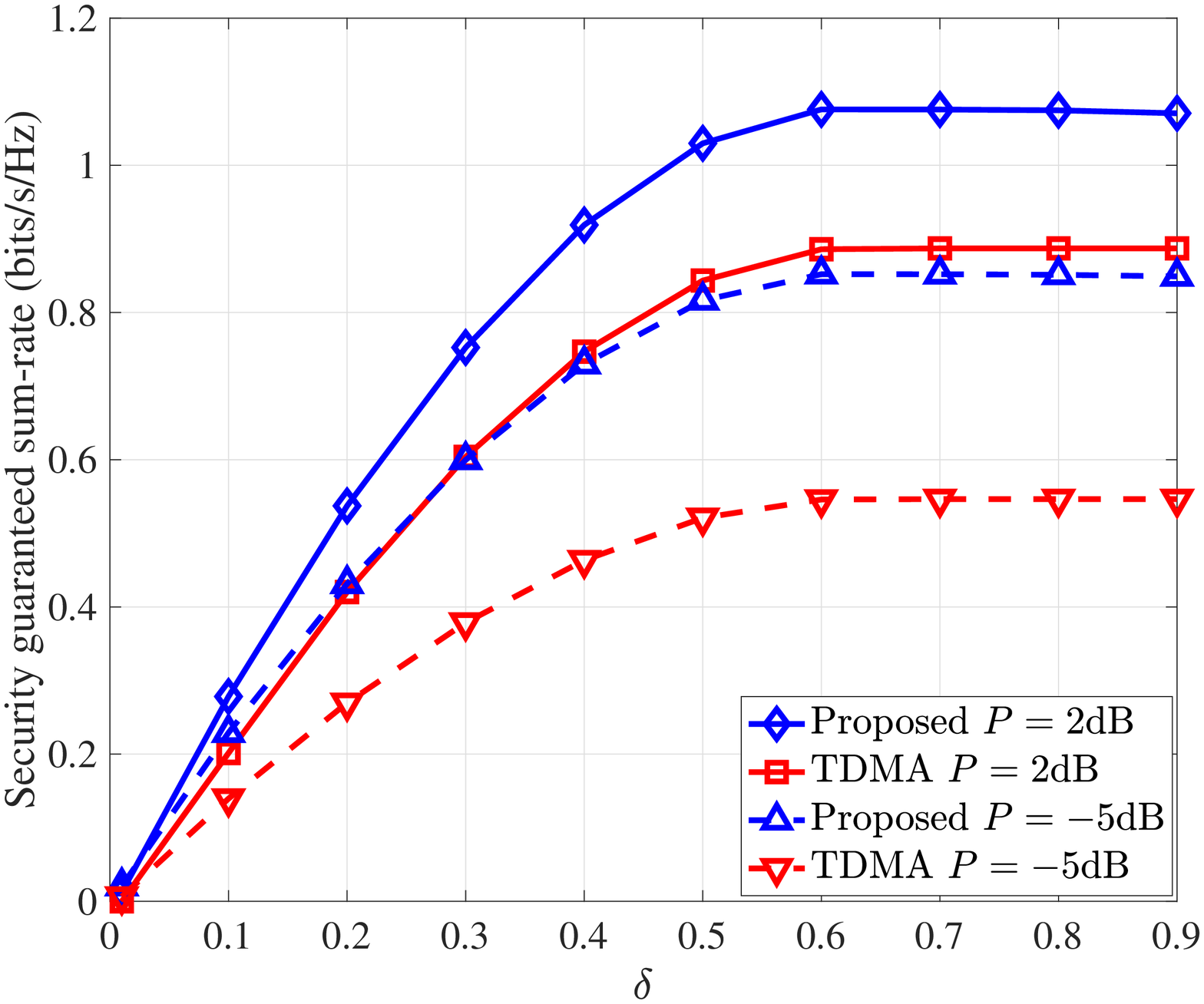}} \hspace{0.0in} 
    \subfigure[]{
        \label{fig:Thr_feedbackBits} %% label for second subfigure
        \includegraphics[width=3in]{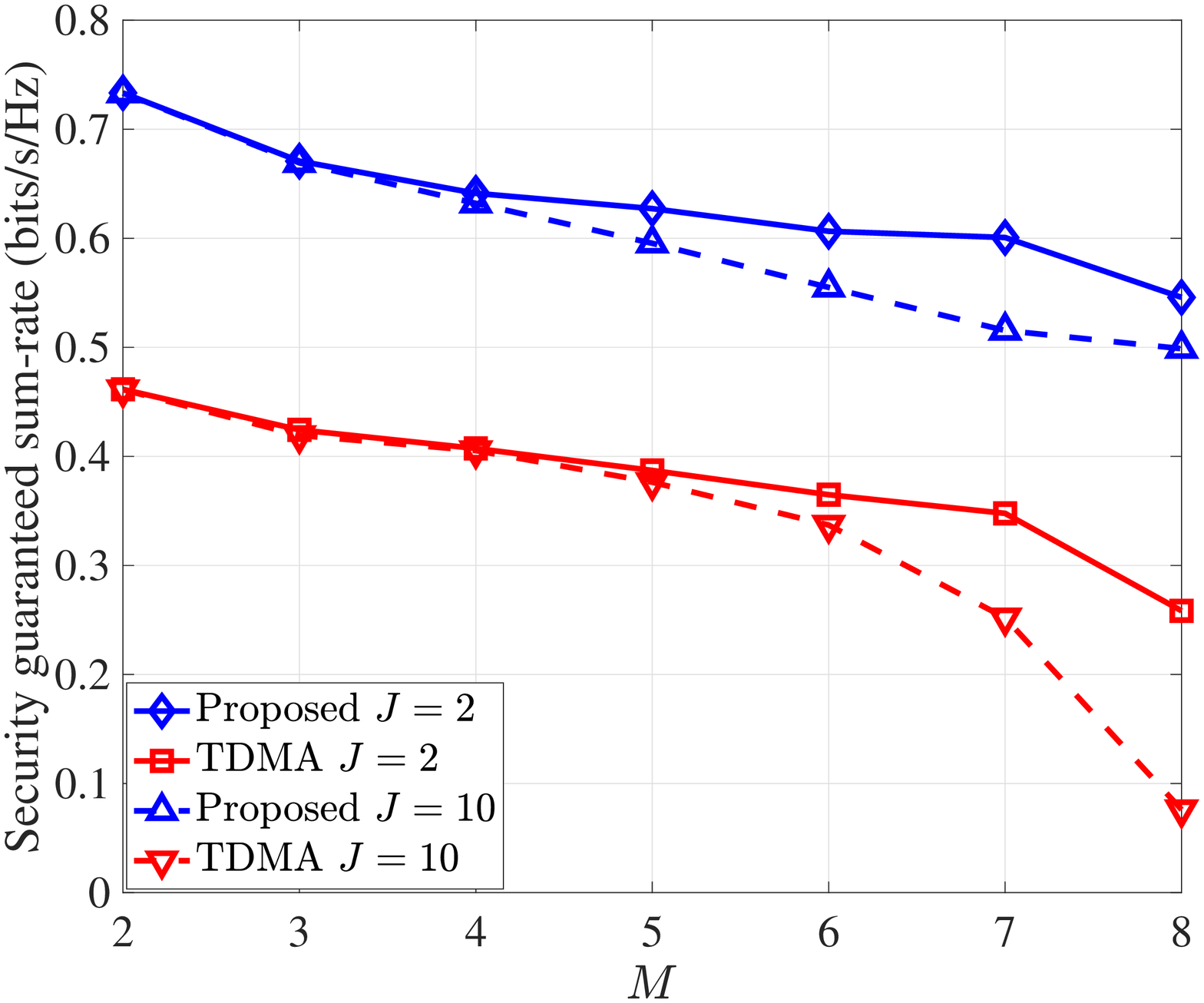}} \hspace{0.0in} 
    \caption{Comparison with TDMA scheme with $N = 100$: (a)~Security guaranteed sum-rates versus $\delta$: $M=8$, $J=5$, $\varepsilon=0.3$.~(b)~Security guaranteed sum-rates versus $M$:  $P=-5\mathrm{dB}$, $\delta=0.3$, $\varepsilon=0.1$.}
\end{figure*} 

\begin{figure*} 
    \centering
    \subfigure[]{ 
        \label{fig:Thr_epsilon}%% label for first subfigure
        \includegraphics[width=3in]{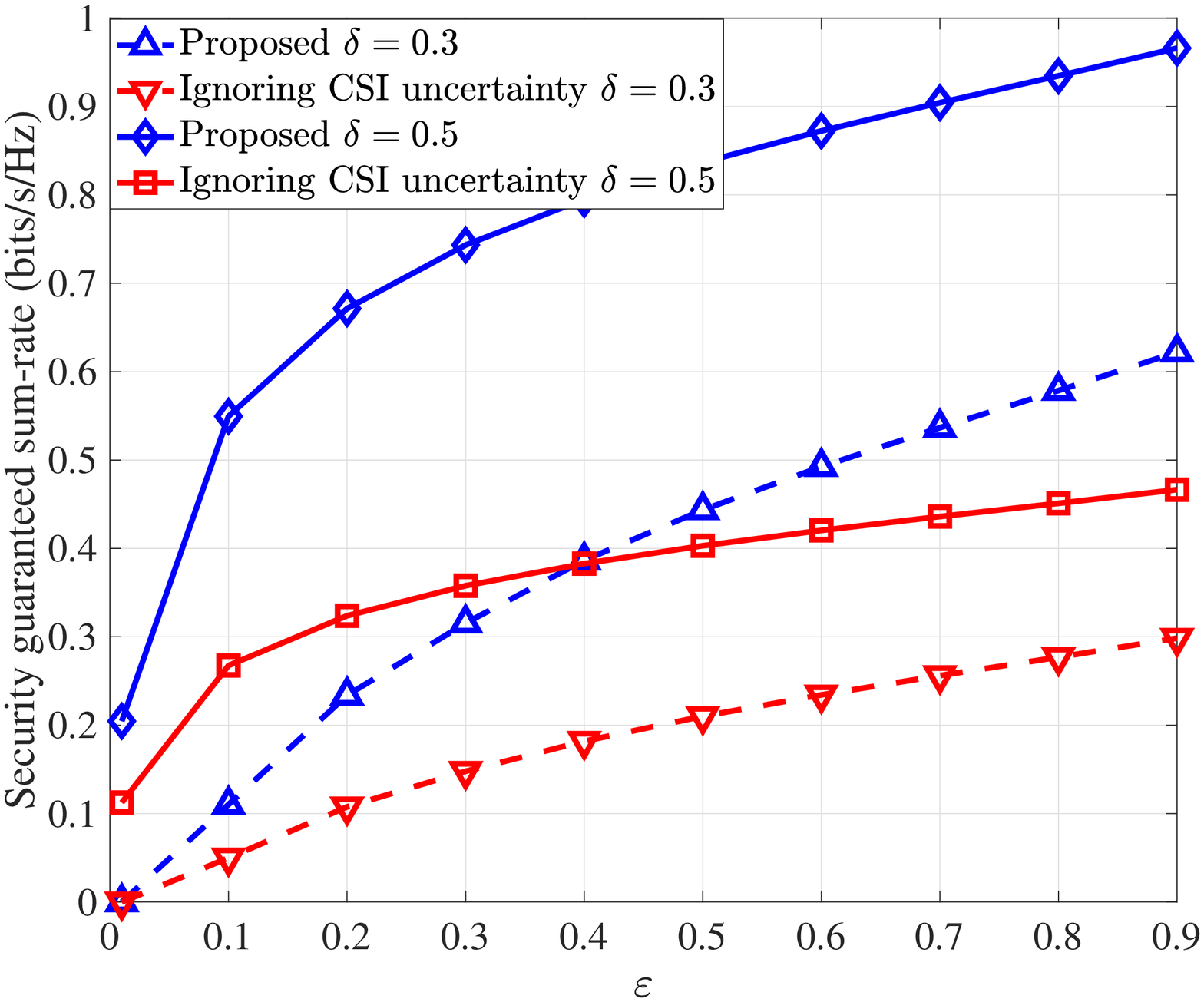}}  \hspace{0.0in} 
    \subfigure[]{ 
        \label{fig:Througput_SNR}%% label for first subfigure
        \includegraphics[width=3in]{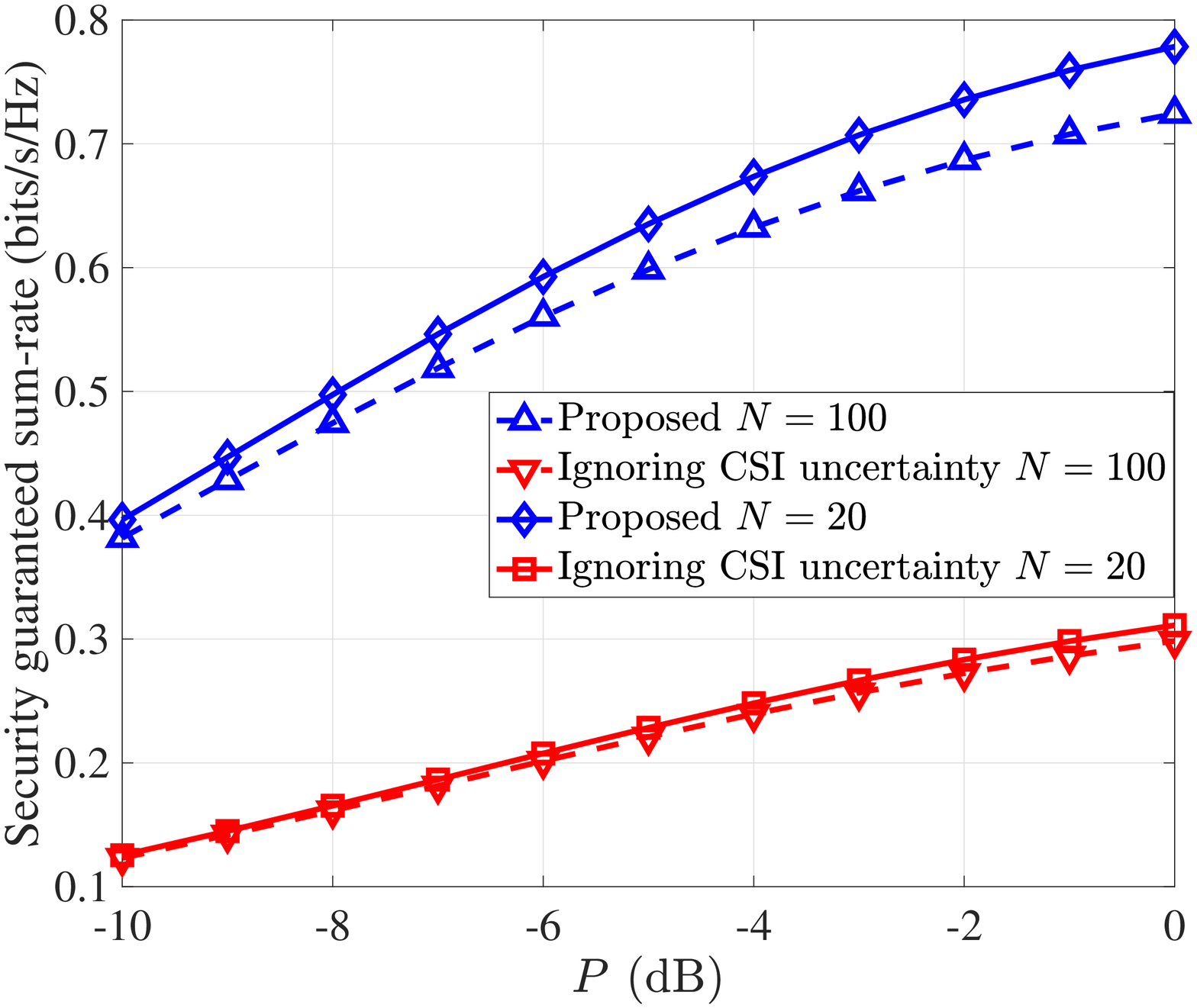}} 
    \caption{Comparison with NOMA ignoring CSI uncertainty~\cite{J_Wang19SecureNOMI} with $M=8$:~(a)~Security guaranteed sum-rates versus $\varepsilon$: $N = 100$, $P=10\mathrm{dB}$.~(b)~Security guaranteed sum-rates versus $P$: $\varepsilon=0.1$, $\delta=0.3$.}  
\end{figure*}

\section{Conclusion}\label{Sec:VII}
This paper studied the secure downlink NOMA transmission under imperfect CSI.
To characterize the performance of this system, an efficient first-order algorithm was proposed to maximize the security guaranteed sum-rate under the constraints of outage probability and transmit power budget. Since the proposed first-order algorithm is Hessian-free, it has a linear complexity order with respect to the number of users in the system, making it ideal for massive access scenarios. 
Numerical results demonstrated that the proposed first-order algorithm achieves identical performance to the conventional method but saves at least two orders of magnitude in computation time, and it significantly improves the security guaranteed sum-rate compared to orthogonal multiple access transmission, and NOMA transmission ignoring CSI uncertainty.

\appendices

\section{Derivation of~\eqref{eq:COP_cf_final}
}\label{ref:der_cop_mk}
Based on~\eqref{eq:COP_cf_inter}, $p^{m,k}_{co}$ can be expressed as
\begin{align}
&p^{m,k}_{co}\nonumber\\
=&\mathrm{Pr}\left\{2^{R_{m,k}}-1\!>\!\frac{\mathcal{X}\theta_{m,k}}{\mathcal{X}\sum\limits_{i=1}^{k-1}\theta_{m,i}+{P_m}\mathcal{Y}+\frac{1}{\gamma_{m,k}}}\right\}  \nonumber\\
=&\mathrm{Pr}\left\{\!\frac{2^{R_{m,k}}-1}{\theta_{m,k}\!-\!\left(2^{R_{m,k}}-1\right)\sum\limits_{i=1}^{k-1}\theta_{m,i}}\!>\!\frac{\mathcal{X}}{{P_m}\mathcal{Y}+\frac{1}{\gamma_{m,k}}}\!\right\} \label{eq:COP_CF_tempt},
\end{align}
where  $\mathcal{X}=|{\mathbf{g}}^H_{m,k}\mathbf{w}_m/\mu_{m,k}|^2\geq 0
$, $\mathcal{Y}=\|{\mathbf{g}}_{m,k}/\mu_{m,k}\|^2\sin^2 \beta_{m,k} \sum_{v\neq m}| \mathbf{e}^H_{m,k} \mathbf{w}_v|^2\geq 0$, and~\eqref{eq:COP_CF_tempt} is due to $\theta_{m,k}-\left(2^{R_{m,k}}-1\right)\sum_{i=1}^{k-1}\theta_{m,i}>0$~\cite[eq. 8]{J_Ding14NOMA}, otherwise, $p^{m,k}_{co}$ is always one.
Furthermore, based on the independent property of the interference terms~\cite[eq. 24]{J_Kout12DownSDMA}, variables of $\mathcal{X}$ and $\mathcal{Y}$ are independent.

To obtain a closed-form expression of $p^{m,k}_{co}$, we first provide the probability density function (PDF) of $\mathcal{X}$. 
Denote $\mathbf{\Psi}_{m,k}=\frac{{\mathbf{g}}_{m,k}}{\mu_{m,k}}\sim \mathcal{CN}(\mathbf{0},\mathbf{I}_N)$ and $\tilde{\mathbf{\Psi}}_{m,k}=\mathbf{\Psi}_{m,k}/\|\mathbf{\Psi}_{m,k}\|$, $\mathcal{X}$ can be rewritten as $\|{\mathbf{\Psi}}_{m,k}\|^2|\tilde{\mathbf{\Psi}}_{m,k}^H\mathbf{w}_m|^2$. 
Since the normalized beamformer $\mathbf{w}_m$ is determined by $\{\hat{\mathbf{g}}_n\}_{n\neq m}$ according to~\eqref{eq:zf_beamformer} and $\{\hat{\mathbf{g}}_n\}_{n\neq m}$ is independent of ${\mathbf{g}}_{m,k}$, ${\mathbf{g}}_{m,k}$ and $\mathbf{w}_m$ are independent. As a result, $\tilde{\mathbf{\Psi}}_{m,k}$ and $\mathbf{w}_m$ are independent and unit vectors in $N$ dimensional space. Based on~\cite[Lemma 1]{J_Roh06Se_Beamform}, the square inner product between two independent unit-norm random vectors $X_1:=|\tilde{\mathbf{\Psi}}^H_{m,k}\mathbf{w}_m|^2$ follows $\mathrm{Beta}(1,N-1)$ and its PDF is 
$f_{X_1}(x_1)=(1-x_1)^{N-2}/Be(1,N-1), x_1\in [0,1]$,
where $Be(x,y)$ is the beta function~\cite[eq. 8.380]{Table_Mathe00}.
On the other hand, since $\mathbf{\Psi}_{m,k}\sim \mathcal{CN}(\mathbf{0},\mathbf{I}_N)$, $X_2:=\|\mathbf{\Psi}_{m,k}\|^2$ follows a $\chi^2$ distribution with $2N$ degrees of freedom, and its PDF is 
$f_{X_2}(x_2)=\frac{x_2^{N-1}e^{-x_2/2}}{2^N\Gamma(N)}, x_2\geq 0$,
where $\Gamma(x)$ is the Gamma function~\cite[eq. 8.310]{Table_Mathe00}. 
Since $\mathcal{X}=X_1X_2$ and $X_1$ and $X_2$ are independent,  
the PDF of $\mathcal{X}$ is given by
\begin{align}\label{eq:pdf-X}
f_{\mathcal{X}}(x)&=\int_{x_2}\frac{1}{|x_2|}f_{X_2}(x_2)f_{X_1}\left(\frac{x}{x_2}\right)\mathrm{d}x_2 \nonumber\\
&=\frac{\int_{x}^{+\infty}(x_2-x)^{N-2}e^{-\frac{x_2}{2}}\mathrm{d}x_2}{Be(1,N-1)2^N\Gamma(N)}\nonumber \\
&=\frac{1}{2}e^{-\frac{x}{2}}, \quad x\geq 0.
\end{align}

Now, we derive the PDF of $\mathcal{Y}$. The cumulative distribution function of $\sin^2 \beta_{m,k}$ is given by~\cite{J_Yoo_pdf07}
\begin{equation}\label{eq:pdf_beta_mk}
\begin{split}
&F\left(\sin^2\! \beta_{m,k}\right)\\
=&\left\{\begin{array}{ll}
\!2^B \left(\sin^2 \!\beta_{m,k}\right)^{N-1}
, & \!\!\mathrm{if} ~0\leq \sin^2 \beta_{m,k} \leq 2^{-\frac{B}{N-1}}, \\
\!1, &\!\!\mathrm{if} ~\sin^2 \beta_{m,k}\geq 2^{-\frac{B}{N-1}}.
\end{array}\right.
\end{split}
\end{equation}
Hence, we have $\|\mathbf{\Psi}_{m,k}\|^2\sin^2 \beta_{m,k} \sim \mathcal{G}\left(N-1,2^{-\frac{B}{N-1}}\right)$ (gamma distribution with shape $N-1$ and scale $2^{-\frac{B}{N-1}}$)~\cite[Lemma 1]{J_Yoo_pdf07}. 
On the other hand, it is known that $ \mathbf{e}_{m,k}$ is a unit vector that has the same distribution as $\tilde{\mathbf{\Psi}}_{m,k}$. Moreover, the unit vector $\mathbf{w}_v$ is isotropic within the $N-1$ dimensional hyperplane and independent of $\mathbf{e}_{m,k}$. Based on~\cite[Lemma 2]{J_Jindal06Finite}, we have $| \mathbf{e}^H_{m,k} \mathbf{w}_v|^2\sim \mathrm{Beta}(1,N-2)$. 
Therefore, by applying~\cite[Lemma 1]{EJ_Zhang09Pro}, $\|\mathbf{\Psi}_{m,k}\|^2\sin^2 \beta_{m,k} |\mathbf{e}^H_{m,k} \mathbf{w}_v|^2 \sim \mathrm{Exp}\left(2^{\frac{B}{N-1}}\right)$.
Since all terms $\{\mathbf{e}^H_{m,k} \mathbf{w}_v\}_{v\neq m}$ are independent of one another, $\mathcal{Y}$ is the sum of $(M-1)$ independent and identically exponentially distributed random variables. Therefore, $\mathcal{Y}\sim \mathcal{G}\left(M-1,2^{-\frac{B}{N-1}}\right)$ and its PDF is expressed as
\begin{equation}\label{eq:pdf-Y}
\begin{split}
f_{\mathcal{Y}}(y)&= \frac{y^{M-2}\exp\left(-y2^{\frac{B}{N-1}}\right)}{2^{-\frac{B(M-1)}{N-1}}\Gamma(M-1)}, \quad y>0.
\end{split}
\end{equation}

Finally, based on~\eqref{eq:pdf-X} and~\eqref{eq:pdf-Y},~\eqref{eq:COP_CF_tempt} is further derived as 
\begin{align}
p^{m,k}_{co}
&=\mathrm{Pr}\left\{\mathcal{X}-{P_m}I\mathcal{Y} <
\frac{I}{\gamma_{m,k}}\right\}\nonumber \\
&=1-e^{-\frac{I}{2\gamma_{m,k}}}
\int_{0}^{\infty}
\frac{y^{M-2}e^{-y\left(2^{\frac{B}{N-1}}+\frac{{P_m}I}{2}\right)}}{2^{-\frac{B(M-1)}{N-1}}\Gamma(M-1)}\mathrm{d}y \nonumber \\
&=1-e^{-\frac{I}{2\gamma_{m,k}}}\left(1+\frac{{P_m}I2^{-\frac{B}{N-1}}}{2}\right)^{-(M-1)}, \label{eq:COP_appendix_CF}
\end{align}
where $I=\frac{2^{R_{m,k}}-1}{\theta_{m,k}-\left(2^{R_{m,k}}-1\right)\sum_{i=1}^{k-1}\theta_{m,i}}>0$, and~\eqref{eq:COP_appendix_CF}  follows from~\cite[eq. 3.326]{Table_Mathe00}.
By putting $I$ into~\eqref{eq:COP_appendix_CF}, $p^{m,k}_{co}$ is obtained as shown in~\eqref{eq:COP_cf_final}.

\section{Proof of Lemma~\ref{lem:Bertin-Type}
}\label{Appd-proof:Lema1}
By putting~\eqref{eq:para_kmj} into~\eqref{ineq:slack_SOP_trans} and re-arranging the terms, we obtain~\eqref{ineq:prox_sop_temp}, shown at the top of the next page,
    \begin{figure*}
	\normalsize
	\begin{equation} \label{ineq:prox_sop_temp}
\begin{split}
&\frac{2^{D^j_{m,k}}-1}{\gamma_{e,j}}\geq 
\left(\theta_{m,k}-\left(2^{D^j_{m,k}}-1\right)\sum_{i\neq k}\theta_{m,i}\right)\mathrm{Tr}\left(\mathbf{W}_m\right)-{P_m}\left(2^{D^j_{m,k}}-1\right)\mathrm{Tr}\left(\mathbf{W}^{\bot}_m\right)\\
&+\sqrt{2\ln(\varepsilon_k^{-1})} \left(\left(\theta_{m,k}-\left(2^{D^j_{m,k}}-1\right)\sum_{i\neq k}\theta_{m,i}\right)\mathrm{Tr}\left(\mathbf{W}_m\right)+{P_m}\left(2^{D^j_{m,k}}-1\right)\|\mathbf{W}^{\bot}_m\|_F \right)\\
&+\ln(\varepsilon_k^{-1})\left(\theta_{m,k}-\left(2^{D^j_{m,k}}-1\right)\sum_{i\neq k}\theta_{m,i}\right)\mathrm{Tr}\left(\mathbf{W}_m\right)
\end{split}
	\end{equation}
	\hrulefill
\end{figure*}
where $\mathbf{W}_m=\mathbf{w}_m\mathbf{w}^H_m$, $\mathbf{W}^{\bot}_m=\sum_{v\neq m}\mathbf{w}_v\mathbf{w}^H_v$.
In order to proceed, we provide the following three facts based on~\eqref{eq:lamda_quadraticForm}. Denote the largest eigenvalue of $\mathbf{\mathbf{\Lambda}}$ by $\lambda_{\max}(\mathbf{\mathbf{\Lambda}})$, we have 
\begin{align}\label{ineq:lamda_sop_f}
&[\lambda_{\max}(\mathbf{\Lambda})]^+\nonumber\\
\leq&\lambda_{\max}\left(\gamma_{e,j}\left(\theta_{m,k}-\left(2^{D^j_{m,k}}-1\right)\sum_{i\neq k}\theta_{m,i}\right)\mathbf{W}_m\right)\nonumber\\
\leq& \gamma_{e,j}\left(\theta_{m,k}-\left(2^{D^j_{m,k}}-1\right)\sum_{i\neq k}\theta_{m,i}\right)\mathrm{Tr}\left(\mathbf{W}_m\right).
\end{align}
On the other hand, we have~\eqref{ineq:SOP_fea_f}, shown at the top of the next page,
    \begin{figure*}
	\normalsize
	\begin{equation} \label{ineq:SOP_fea_f}
\begin{split}
\|\mathbf{\Lambda}\|_F
&\overset{\text{(a)}}{\leq} \gamma_{e,j}\left(\left(\theta_{m,k}-\left(2^{D^j_{m,k}}-1\right)\sum_{i\neq k}\theta_{m,i}\right)\|\mathbf{W}_m\|_F
+{P_m}\left(2^{D^j_{m,k}}-1\right)\|\mathbf{W}^{\bot}_m\|_F \right)\\
&\overset{\text{(b)}} {\leq}\gamma_{e,j} \left(\left(\theta_{m,k}-\left(2^{D^j_{m,k}}-1\right)\sum_{i\neq k}\theta_{m,i}\right)\mathrm{Tr}\left(\mathbf{W}_m\right)+{P_m}\left(2^{D^j_{m,k}}-1\right)\|\mathbf{W}^{\bot}_m\|_F \right)
\end{split}
	\end{equation}
	\hrulefill
\end{figure*}
where step (a) follows from the triangle inequality of the norm, and step (b) follows from $\|\mathbf{W}_m\|_F=\sqrt{\sum_{i=1}^{N}r^2_{i}}\leq \sum_{i=1}^{N}r_{i} = \mathrm{Tr}\left(\mathbf{W}_m\right)$ with $\{r_{i}\geq 0\}_{i=1}^N$ being the eigenvalues of $\mathbf{W}_m$.
Furthermore, $\mathrm{Tr}\left(\mathbf{\Lambda}\right)$ is expressed as
\begin{align}\label{eq:tr_SOP_aux}
\mathrm{Tr}\left(\mathbf{\Lambda}\right)=&\gamma_{e,j}\left(\theta_{m,k}-\left(2^{D^j_{m,k}}-1\right)\sum_{i\neq k}\theta_{m,i}\right)\mathrm{Tr}\left(\mathbf{W}_m\right)\nonumber \\
&-\gamma_{e,j}{P_m}\left(2^{D^j_{m,k}}-1\right)\mathrm{Tr}\left(\mathbf{W}^{\bot}_m\right).
\end{align}
Applying~\eqref{ineq:lamda_sop_f}-\eqref{eq:tr_SOP_aux} to~\eqref{ineq:prox_sop_temp}, we obtain
\begin{equation}\label{ineq:SOP_BIT}
\begin{split}
&2^{D^j_{m,k}}-1\geq\\
&\mathrm{Tr}\left(\mathbf{\Lambda}\right)+\sqrt{2\ln(\varepsilon_k^{-1})}\|\mathbf{\Lambda}\|_F+\ln(\varepsilon_k^{-1})[\lambda_{\max}(\mathbf{\Lambda})]^+.
\end{split}
\end{equation}
Comparing both sides of~\eqref{ineq:SOP_BIT} with $\frac{{\mathbf{g}}_{e,j}^H}{\mu_{e,j}}\mathbf{\Lambda}\frac{\mathbf{g}_{e,j}}{\mu_{e,j}}$ and taking probability, we have 
\begin{equation}\label{ineq:Cond_relation}
\begin{aligned}
&\mathrm{Pr}\left\{\frac{{\mathbf{g}}_{e,j}^H}{\mu_{e,j}}\mathbf{\Lambda}\frac{\mathbf{g}_{e,j}}{\mu_{e,j}}> 2^{D^j_{m,k}}-1\right\}\leq  \mathrm{Pr}\Bigg\{\frac{{\mathbf{g}}_{e,j}^H}{\mu_{e,j}}\mathbf{\Lambda}\frac{\mathbf{g}_{e,j}}{\mu_{e,j}}\geq \\
& \mathrm{Tr}(\mathbf{\mathbf{\Lambda}})+\sqrt{2\ln(\varepsilon_k^{-1})}\|\mathbf{\mathbf{\Lambda}}\|_F+\ln(\varepsilon_k^{-1}) [\lambda_{\max}(\mathbf{\mathbf{\Lambda}})]^+\Bigg\}.
\end{aligned}
\end{equation}

On the other hand, since  ${\mathbf{g}}_{e,j}/{\mu_{e,j}}\sim \mathcal{CN}(\mathbf{0},\mathbf{I}_N)$, together with the Hermitian matrix $\mathbf{\Lambda}\in \mathbb{C}^{N\times N}$, for any $\epsilon\geq 0$, we have 
\begin{align}\label{ineq:BTI_Pro_int}
&\mathrm{Pr}\left\{\frac{{\mathbf{g}}_{e,j}^H}{\mu_{e,j}}\mathbf{\Lambda}\frac{\mathbf{g}_{e,j}}{\mu_{e,j}}\geq \mathrm{Tr}(\mathbf{\mathbf{\Lambda}})+\sqrt{2\epsilon}\|\mathbf{\mathbf{\Lambda}}\|_F+\epsilon [\lambda_{\max}(\mathbf{\mathbf{\Lambda}})]^+\right\}\nonumber\\
&\leq \exp(-\epsilon),
\end{align}
which is the Bernstein-Type Inequality (BTI)~\cite{J_Bech09Berntein} and always holds. By substituting $\epsilon=\ln(\varepsilon_k^{-1})$ into~\eqref{ineq:BTI_Pro_int}, we obtain
\begin{equation}\label{ineq:BTI}
\begin{aligned}
\mathrm{Pr}\Bigg\{\frac{{\mathbf{g}}_{e,j}^H}{\mu_{e,j}}\mathbf{\Lambda}\frac{\mathbf{g}_{e,j}}{\mu_{e,j}}\geq & \mathrm{Tr}(\mathbf{\mathbf{\Lambda}})+\sqrt{2\ln(\varepsilon_k^{-1})}\|\mathbf{\mathbf{\Lambda}}\|_F \\
&+\ln(\varepsilon_k^{-1}) [\lambda_{\max}(\mathbf{\mathbf{\Lambda}})]^+\Bigg\}\leq \varepsilon_k
\end{aligned}
\end{equation}
for any $\varepsilon_k\in(0,1]$. Substituting~\eqref{ineq:BTI} into~\eqref{ineq:Cond_relation}, we obtain 
\begin{equation}\label{eq:dir_SOP_BTI}
\mathrm{Pr}\left\{\frac{{\mathbf{g}}_{e,j}^H}{\mu_{e,j}}\mathbf{\Lambda}\frac{\mathbf{g}_{e,j}}{\mu_{e,j}}> 2^{D^j_{m,k}}-1\right\}\leq \varepsilon_k.
\end{equation}
As a result, $\mathrm{Pr}\left\{\frac{{\mathbf{g}}_{e,j}^H}{\mu_{e,j}}\mathbf{\Lambda}\frac{\mathbf{g}_{e,j}}{\mu_{e,j}}> 2^{D^j_{m,k}}-1\right\}\leq \varepsilon$ holds when we set $\varepsilon_k=\varepsilon$.
Applying the result of~\cite[eq. 30]{J_IQF_GausianYI16}, it can be shown that $\mathrm{Pr}\left\{\frac{{\mathbf{g}}_{e,j}^H}{\mu_{e,j}}\mathbf{\Lambda}\frac{\mathbf{g}_{e,j}}{\mu_{e,j}}> 2^{D^j_{m,k}}-1\right\}\leq \varepsilon$ is equivalent to~\eqref{ineq:cons_SOP_CF}.
Therefore, if~\eqref{ineq:slack_SOP_trans} holds for any $\varepsilon_k\in(0,1]$, then~\eqref{ineq:cons_SOP_CF} holds, i.e.,~\eqref{ineq:slack_SOP_trans} is a tighter constraint than the SOP constraint of~\eqref{ineq:cons_SOP_CF}.

\section{Proof of Lemma~\ref{lem:clo-fom_zub} 
}\label{proof:lem-cf_z_ub}

Since by definition~\eqref{eq:aux_int_xi}, $\xi_{m,k}\geq 0$. Hence, we only need to find the upper bound of $\xi_{m,k}$ to determine its feasible set.
Denote $q\left(\xi_{m,k}\right)$ equals to the left hand side of~\eqref{ineq:COP_aux_P1}, the first-order derivative of $q\left(\xi_{m,k}\right)$ is given by
\begin{equation}
\begin{split}
&q'\left(\xi_{m,k}\right)
=-\frac{\left(\frac{2+\xi_{m,k}{P_m}2^{-\frac{B}{N-1}}}{4\gamma_{m,k}}+\frac{(M-1){P_m}2^{-\frac{B}{N-1}}}{2}\right)}{\exp\left(\frac{\xi_{m,k}}{2\gamma_{m,k}}\right)\left(1+\frac{\xi_{m,k}{P_m}2^{-\frac{B}{N-1}}}{2}\right)^{M}},
\end{split}
\end{equation}
which is negative for $M\geq 1$. With the decreasing property of $q\left(\xi_{m,k}\right)$ and lower boundedness of $q\left(\xi_{m,k}\right)$ in~\eqref{ineq:COP_aux_P1}, an upper bound of $\xi_{m,k}$, denoted by $\xi^{ub}_{m,k}$, is obtained by solving the following equation 
\begin{equation}\label{eq:inter_CF_ub_xi}
\exp\left(-\frac{\xi^{ub}_{m,k}}{2\gamma_{m,k}}\right)
\left(1+\xi^{ub}_{m,k}\frac{{P_m}2^{-\frac{B}{N-1}}}{2}\right)^{1-M}=1-\delta.
\end{equation}
By straightforward algebra,~\eqref{eq:inter_CF_ub_xi} can be further re-expressed as
\begin{equation}\label{eq:inter_z-ub}
\begin{split}
\frac{2^{\frac{B}{N-1}}+\frac{{P_m}\xi^{ub}_{m,k}}{2}}{{P_m}\gamma_{m,k}(M-1)}&
\exp\left(\frac{2^{\frac{B}{N-1}}+\frac{{P_m}\xi^{ub}_{m,k}}{2}}{{P_m}\gamma_{m,k}(M-1)}\right)\\
&=\frac{2^{\frac{B}{N-1}}\exp\left(\frac{2^{\frac{B}{N-1}}}{\gamma_{m,k}(M-1){P_m}}\right)}{\gamma_{m,k}(M-1){P_m}(1-\delta)^{\frac{1}{M-1}}}.
\end{split}
\end{equation} 
With the help of the principal branch of Lambert W function~\cite{Lambert_W96},~\eqref{eq:inter_z-ub} is rewritten as
\begin{equation}\label{eq:W0-up-Xi}
W_0\left(\frac{2^{\frac{B}{N-1}}\exp\left(\frac{2^{\frac{B}{N-1}}}{\gamma_{m,k}(M-1){P_m}}\right)}{\gamma_{m,k}(M-1){P_m}(1-\delta)^{\frac{1}{M-1}}}\right)=\frac{2^{\frac{B}{N-1}}+\frac{{P_m}\xi^{ub}_{m,k}}{2}}{{P_m}\gamma_{m,k}(M-1)}.
\end{equation}
Re-arranging the terms in~\eqref{eq:W0-up-Xi} leads to~\eqref{eq:up_xi_CF}.

\section{Proof of Theorem~\ref{tem:FP_CCP_def} 
}\label{proof:them-FP_condi}
Firstly, since the feasible set of $\mathcal{Q}1$ is determined by~\eqref{ineq:set_xi_Q1} with simple bound constraints, it must be a nonempty standard convex set. 

Secondly, to prove the concavity of $A_k(\Xi^{(m,j)}_{k})$, we rewrite $A_k(\Xi^{(m,j)}_{k})$ as
$A_k(\Xi^{(m,j)}_{k})=\left[\log_2(1+E(\Xi^{(m,j)}_{k}))-
\log_2 \left(1+ \frac{\Theta^{(m,j)}_{k}}{ \kappa_{m,k,j} + \sum_{i\neq k}\Theta_{m,i}}\right)
\right]^+$,
where 
$E(\Xi^{(m,j)}_{k})=\frac{\Xi^{(m,j)}_{k}\Theta^{(m,j)}_{k}}{1+\Xi^{(m,j)}_{k}\sum_{i=1}^{k-1}\Theta^{(m,j)}_{i}}\geq 0$. The first-order and second-order derivatives of $E(\Xi^{(m,j)}_{k})$ are respectively given by
\begin{equation}\label{eq:der_zk_xi-k}
E'(\Xi^{(m,j)}_{k})    =\frac{\Theta_{m,k}}{\left(1+\Xi^{(m,j)}_{k}\sum\limits_{i=1}^{k-1}\Theta^{(m,j)}_{i}\right)^2}\geq 0,
\end{equation}
\begin{equation}
E''(\Xi^{(m,j)}_{k})    =\frac{-2\Theta^{(m,j)}_{k}\sum\limits_{i=1}^{k-1}\Theta^{(m,j)}_{i}}{\left(1+\Xi^{(m,j)}_{k}\sum\limits_{i=1}^{k-1}\Theta^{(m,j)}_{i}\right)^3}\leq 0.
\end{equation}
Since $E''(\Xi^{(m,j)}_{k})\leq 0$ for $0 \leq  \Xi^{(m,j)}_{k} \leq \xi^{ub}_{m,k}$, $E(\Xi^{(m,j)}_{k})$ is concave on $\Xi^{(m,j)}_{k}$. Furthermore, due to the concavity and non-decreasing property of function $\log_2(1+x)$ with $x>0$, $\log_2(1+E(\Xi^{(m,j)}_{k}))-\log_2\left(1+\frac{\Theta^{(m,j)}_{k}}{\kappa_{m,k,j} + \sum_{i\neq k}\Theta^{(m,j)}_{i}}\right)$ is concave on $\Xi^{(m,j)}_{k}$. 
Since pointwise maximum operation preserves concavity~\cite{Cov_Opt90}, $A_k(\Xi^{(m,j)}_{k})$ is concave on $\Xi^{(m,j)}_{k}$.

Thirdly, to prove the convexity of $B_k(\Xi^{(m,j)}_{k})$, we rewrite $B_k(\Xi^{(m,j)}_{k})$ as $\hat{b}_k(\Xi^{(m,j)}_{k})\tilde{b}_k(\Xi^{(m,j)}_{k})$, where $\hat{b}_k(\Xi^{(m,j)}_{k}):=\exp\left(\frac{\Xi^{(m,j)}_{k}}{2\gamma_{m,k}}\right)$ and
$\tilde{b}_k(\Xi^{(m,j)}_{k}):=\left(1+\Xi^{(m,j)}_{k}\frac{P_m}{2^{\frac{B}{N-1}+1}}\right)^{M-1}$.
It is obvious that $\hat{b}_k(\Xi^{(m,j)}_{k})$ is convex on $\Xi^{(m,j)}_{k}$. On the other hand, since $\tilde{b}_k(\Xi^{(m,j)}_{k})$ is the composition with an affine mapping from the convex function $x^{M-1}$ with $x\geq 0$, $\tilde{b}_k(\Xi^{(m,j)}_{k})$ has the same convex property as $x^{M-1}$. Considering that $\hat{b}_k(\Xi^{(m,j)}_{k})$ and $\tilde{b}_k(\Xi^{(m,j)}_{k})$ are both convex, $B_k(\Xi^{(m,j)}_{k})$ is convex since convexity is closed under multiplication and positive scaling~\cite{Cov_Opt90}.

\section{Proof of Proposition~\ref{lem:sta_xi_mk_fp}
}\label{prrof:lema_opt_Q2}
The maximizer for $\mathcal{Q}2$ would either be at the stationary point ${\Xi^{(m,j)}_{k}}^{\diamond}$ or boundary points of the feasible range $[0, \xi^{ub}_{m,k}]$. Since $g(0,y^{\dagger}_k)=0$ and the objective function of $\mathcal{Q}2$ must be non-negative at optimality, the optimal ${\Xi^{(m,j)}_{k}}^{\dagger}$ cannot be 0. As a result, the optimal ${\Xi^{(m,j)}_{k}}^{\dagger}$ is either ${\Xi^{(m,j)}_{k}}^{\diamond}$ or $ \xi^{ub}_{m,k}$. On the other hand, if ${\Xi^{(m,j)}_{k}}^{\diamond}>\xi^{ub}_{m,k}$, then ${\Xi^{(m,j)}_{k}}^{\dagger}=\xi^{ub}_{m,k}$ since the optimal ${\Xi^{(m,j)}_{k}}^{\dagger}\leq \xi^{ub}_{m,k}$. Otherwise, due to the concavity of $g(\Xi^{(m,j)}_{k},y^{\dagger}_k)$ on $\Xi^{(m,j)}_{k}$, ${\Xi^{(m,j)}_{k}}^{\dagger}={\Xi^{(m,j)}_{k}}^{\diamond}$. 
Therefore, the optimal  ${\Xi^{(m,j)}_{k}}^{\dagger}$ is obtained as shown in~\eqref{eq:opt_Xi-Q2}.
    
Next, we determine the stationary point ${\Xi^{(m,j)}_{k}}^{\diamond}$. The pointwise-maximum function $A_k(\Xi^{(m,j)}_{k})$ can be rewritten as 
$A_k(\Xi^{(m,j)}_{k})=\tilde{A}_k(\Xi^{(m,j)}_{k})\mathbb{I}\left(\Xi^{(m,j)}_{k}>\left(\kappa_{m,k,j} + \sum_{i= k+1}\Theta^{(m,j)}_{i}\right)^{-1}\right)$, where $\tilde{A}_k(\Xi^{(m,j)}_{k})$ is given by
    \begin{equation}
    \begin{split}
    \tilde{A}_k(\Xi^{(m,j)}_{k})=&
    \log_2\left(\frac{1+\frac{\Xi^{(m,j)}_{k}\Theta^{(m,j)}_{k}}{1+\Xi^{(m,j)}_{k}\sum\limits_{i=1}^{k-1}\Theta^{(m,j)}_{i}}}{1+\frac{\Theta^{(m,j)}_{k}}{\kappa_{m,k,j} + \sum\limits_{i\neq k}\Theta^{(m,j)}_{i}}}
    \right),
    \end{split}
    \end{equation}
    and $\mathbb{I}(H)$ is the indicator function with $\mathbb{I}(H)=1$ if the event \textit{H} occurs and $\mathbb{I}(H)=0$ otherwise. 
    It is known that if $A_k(\Xi^{(m,j)}_{k})=0$, $g(\Xi^{(m,j)}_{k},y^{\dagger}_k)=0$ and the objective function of $ \mathcal{Q}2$ is always 0. Hence, substituting $A_k(\Xi^{(m,j)}_{k})=\tilde{A}_k(\Xi^{(m,j)}_{k})$ into $g(\Xi^{(m,j)}_{k},y^{\dagger}_k)$, 
    the stationary point ${\Xi^{(m,j)}_{k}}^{\diamond}$ is the unique root for $\frac{\partial g(\Xi^{(m,j)}_{k},y^{\dagger}_k)}{\partial\Xi^{(m,j)}_{k}}=0$, which is equivalent to~\eqref{eq:der_1st_bis}, shown at the top of this page.
    \begin{figure*}
    	\normalsize
    	\begin{equation} \label{eq:der_1st_bis}
   \frac{\Theta^{(m,j)}_{k}\left(1+{\Xi^{(m,j)}_{k}}^{\diamond}\sum\limits_{i=1}^{k-1}\Theta^{(m,j)}_{i}\right)^{-1}\tilde{A}^{-1/2}_k({\Xi^{(m,j)}_{k}}^{\diamond})}{\left(1+{\Xi^{(m,j)}_{k}}^{\diamond}\sum\limits_{i=1}^{k}\Theta^{(m,j)}_{i}\right)B_k({\Xi^{(m,j)}_{k}}^{\diamond})\ln 2}\!=\!y_k^{\dagger} \left(\frac{1}{2\gamma_{m,k}}+\frac{(M-1){P_m}2^{-\frac{B}{N-1}}}{2+{P_m}{\Xi^{(m,j)}_{k}}^{\diamond}2^{-\frac{B}{N-1}}}\right)
    	\end{equation}
    	\hrulefill
    \end{figure*}
Re-arranging the terms in~\eqref{eq:der_1st_bis} lead to~\eqref{eq:sta_xi_mk}.

\section{Derivation of the gradient of $f(\{\Theta^{(m,j)}_{k}\}_{k=1}^{K_m})$
}\label{lem:CF_theta_AM}
From~\eqref{D2: obj_theta}, $f(\{\Theta^{(m,j)}_{k}\}_{k=1}^{K_m})$ can be rewritten as 
\begin{align} \label{eq:f_rewrite_sumform}
&f(\{\Theta^{(m,j)}_{k}\}_{k=1}^{K_m})\nonumber\\
=&\underbrace{\log_2\left(1+\Xi^{(m,j)}_{k}\sum_{i=1}^{k}\Theta^{(m,j)}_{i}\right)}_{:= f_1(\{\Theta^{(m,j)}_{k}\}_{k=1}^{K_m})}-\underbrace{\log_2\left(1+\Xi^{(m,j)}_{k}\sum_{i=1}^{k-1}\Theta^{(m,j)}_{i}\right)}_{:= f_2(\{\Theta^{(m,j)}_{k}\}_{k=1}^{K_m})}\nonumber \\
&+\underbrace{\log_2\left(\kappa_{m,k,j} + \sum_{i\neq k}\Theta^{(m,j)}_{i}\right)}_{:= f_3(\{\Theta^{(m,j)}_{k}\}_{k=1}^{K_m})}-\log_2\left(\kappa_{m,k,j} +{P_m}\right).
\end{align}
Then, the gradient of $f_1$, $f_2$ and $f_3$ with respect to $\Theta^{(m,j)}_{i}, i\in \{1,\ldots,K_m\}$ are respectively derived as
$\nabla_{\Theta^{(m,j)}_{i}} f_1= \frac{1}{\ln2}\frac{\Xi^{(m,j)}_{k}\mathbb{I}\left(i\leq k\right)}{1+\Xi^{(m,j)}_{k}\sum_{i=1}^{k}\Theta^{(m,j)}_{i}}$, $\nabla_{\Theta^{(m,j)}_{i}} f_2=\frac{1}{\ln2}\frac{\Xi^{(m,j)}_{k}\mathbb{I}\left(i< k\right)}{1+\Xi^{(m,j)}_{k}\sum_{i=1}^{k-1}\Theta^{(m,j)}_{i}}$, and $\nabla_{\Theta^{(m,j)}_{i}} f_3=\frac{1}{\ln2}\frac{\mathbb{I}\left( i\neq k \right)}{\kappa_{m,k,j} + \sum_{i\neq k}\Theta^{(m,j)}_{i}}$.
Therefore, the gradient of $f(\{\Theta^{(m,j)}_{k}\}_{k=1}^{K_m})$ with respect to $\Theta^{(m,j)}_{i}, i\in \{1,\ldots,K_m\}$ is given by~\eqref{eq:clo_explicit_dif}, shown at the top of the next page.
    \begin{figure*}
	\normalsize
	\begin{equation} \label{eq:clo_explicit_dif}
\begin{split}
\nabla_{\Theta^{(m,j)}_{i}} f=\left\{\begin{array}{ll}
\frac{1}{\ln2}\left(  \frac{\Xi^{(m,j)}_{k}}{1+\Xi^{(m,j)}_{k}\sum\limits_{i=1}^{k}\Theta^{(m,j)}_{i}}- \frac{\Xi^{(m,j)}_{k}}{1+\Xi^{(m,j)}_{k}\sum\limits_{i=1}^{k-1}\Theta^{(m,j)}_{i}}+ \frac{1}{\kappa_{m,k,j} + \sum\limits_{i\neq k}\Theta^{(m,j)}_{i}}  \right), \!& \mathrm{if}~ i<k, \\
\frac{1}{\ln2}\frac{\Xi^{(m,j)}_{k}}{1+\Xi^{(m,j)}_{k}\sum\limits_{i=1}^{k}\Theta^{(m,j)}_{i}}, \!&\mathrm{if}~i= k, \\
\frac{1}{\ln2}\frac{1}{\kappa_{m,k,j} + \sum\limits_{i\neq k}\Theta^{(m,j)}_{i}}, \!&\mathrm{if}~i> k
\end{array}\right.
\end{split}
	\end{equation}
	\hrulefill
\end{figure*}

\section{Proof of Proposition~\ref{lem:Proj_gra_theta_k}
}\label{proof:lemma_proj_gra_theta}
The Lagrangian function of~\eqref{opt:proj_D2} is given by
    \begin{align}\label{eq:lagar_dualFunc}
    \mathcal{L}\left(\{\Theta^{(m,j)}_{k}\}_{k=1}^{K_m},\zeta\right)=&    \sum_{k=1}^{K_m}\left\|\Theta^{(m,j)}_{k}-{\Theta^{(m,j)}_{k}}^{\diamond}\left(l+\frac{1}{2}\right)\right\|^2\nonumber \\
&+\zeta \left(\sum_{k=1}^{K_m}\Theta^{(m,j)}_{k}-{P_m}\right),
    \end{align}
    where $\zeta$ is the dual variable corresponding to the constraint in~\eqref{D2: obj_theta}.
    Based on feasible set $\mathcal{P}_{\mathcal{D}_2}$, $\Theta^{(m,j)}_{k}$ can be either 0 or positive. If $\Theta^{(m,j)}_{k}>0$, the optimal solution must satisfy the following KKT conditions: ${\Theta^{(m,j)}_{k}}^{\dagger}-{\Theta^{(m,j)}_{k}}^{\diamond}(l+1/2)+\zeta^{\dagger}/2=0$ and
    $\sum_{k=1}^{K_m}{\Theta^{(m,j)}_{k}}^{\dagger}-{P_m}=0$.
    Therefore, ${\Theta^{(m,j)}_{k}}^{\dagger}$ is derived as
    ${\Theta^{(m,j)}_{k}}^{\dagger}={\Theta^{(m,j)}_{k}}^{\diamond}(l+1/2)-\zeta^{\dagger}/2$,
    where the optimal $\zeta^{\dagger}$ is given by
    \begin{equation}
    \zeta^{\dagger}=\frac{2}{K_m}\left(\sum_{k=1}^{K_m}{\Theta^{(m,j)}_{k}}^{\diamond}\left(l+\frac{1}{2}\right)-{P_m}\right).
    \end{equation}
    Together with the case of $\Theta^{(m,j)}_{k}=0$, the optimal solution to~\eqref{opt:proj_D2} is shown in~\eqref{eq:pro_gra_optTheta}.

\section{Proof of Theorem~\ref{tem:conver_AM}
}\label{proof_tem_convergence}
Define $\mathbf{Z}=(\mathbf{\Xi},\boldsymbol{\Theta})\in \mathcal{Z}$ as the composite vector with $\mathbf{\Xi}:=\{\Xi^{(m,j)}_{k}\}_{k=1}^{K_m}$ and $\boldsymbol{\Theta}:=\{\Theta^{(m,j)}_{k}\}_{k=1}^{K_m}$, where $\mathcal{Z}$ is the feasible set of $\mathcal{P}2^{[m,j]}$. Furthermore, the $l^{th}$ AM iteration is denoted by $\mathbf{Z}^{(l)}=(\mathbf{\Xi}^{(l)},\boldsymbol{\Theta}^{(l)})$, the sequence in between is denoted by $\mathbf{Z}^{(l+\frac{1}{2})}=(\mathbf{\Xi}^{(l+1)},\boldsymbol{\Theta}^{(l)})$, and the objective function of $\mathcal{P}2^{[m,j]}$ is denoted by $\Upsilon(\mathbf{Z})$. First, we show that the sequence of solutions converges to a limit point $\mathbf{Z}^*$. 

\begin{lemma}\label{lem:limit_point_property}
The sequence of solutions $\{\mathbf{Z}^{(l)}\}_{l\in \mathbb{N}}$ generated by AM iteration is bounded and must have a limit point $\mathbf{Z}^*$.
\end{lemma}
\begin{proof}
We first prove that $\Upsilon(\mathbf{Z})$ is monotonically increasing as iteration number increases. 
It is known that for the $l^{th}$ iteration, to update $\mathbf{\Xi}^{(l)}$ with Algorithm~\ref{alg:CDM_iterative_RI}, the obtained point $\mathbf{\Xi}^{(l+1)}$ is a local optimal point, and it must be a saddle-point for $\mathcal{P}2^{[m,j]}$. Together with the property of saddle-points~\cite{J_Liu09Saddle}, we have the following inequality
\begin{equation}\label{ineq:AM_conv_B2}
\Upsilon(\mathbf{Z}^{(l+\frac{1}{2})})\geq 
\Upsilon(\mathbf{Z}^{(l)}).
\end{equation}
On the other hand, with Algorithm~\ref{alg:dual_update_theta}, the obtained  $\boldsymbol{\Theta}^{(l+1)}$ is a stationary point. Furthermore, since the feasible set of $\mathcal{P}2^{[m,j]}$ is a Cartesian product of convex sets, the optimization over $\boldsymbol{\Theta}$ is independent on $\mathbf{\Xi}$. Hence, we have the following inequality
\begin{equation}\label{ineq:Converge_AM_condi}
\Upsilon(\mathbf{Z}^{(l+1)})\geq \Upsilon(\mathbf{Z}^{(l+\frac{1}{2})}).
\end{equation}
Combining~\eqref{ineq:AM_conv_B2} and~\eqref{ineq:Converge_AM_condi}, we conclude that 
\begin{equation}\label{eq:mon_incre_pro1}
\Upsilon(\mathbf{Z}^{(l+1)}) \geq \Upsilon(\mathbf{Z}^{(l)}) \geq \cdots\geq  \Upsilon(\mathbf{Z}^{(0)}),
 \forall l\in\{1,2\ldots\},
\end{equation}
where $\Upsilon(\mathbf{Z}^{(0)})$ is any finite initial value of the objective function. 

Then we prove the boundedness of the sequence of solutions $\{\mathbf{Z}^{(l)}\}_{l\in \mathbb{N}}$ generated by AM iteration. It is observed that $\mathbf{\Xi}$ and $\boldsymbol{\Theta}$ are respectively located in separable closed sets based on~\eqref{ineq:cons_xi_orthant} and~\eqref{eq:simplex_theta}. Hence, the sequence $\{\mathbf{Z}^{(l)}\}_{l\in \mathbb{N}}$ is bounded. 
Together with monotonic property of~\eqref{eq:mon_incre_pro1}, $\{\mathbf{Z}^{(l)}\}_{l\in \mathbb{N}}$ must have a limit point $\mathbf{Z}^*$ based on Bolzano-Weierstrass theorem~\cite{B_BartleRobert11}.
\end{proof}

To further investigate the property of limit point $\mathbf{Z}^*$, we first recall the notion of gradient mapping.
The gradient mappings with respect to $\mathbf{Z}\in \mathcal{Z}$ for any $L>0$ is defined as~\cite{B_Nesterov14}
\begin{equation}\label{eq:gra_map_def}
G_L(\mathbf{Z})=L\left(\mathbf{Z}-\mathrm{prox}_L^{\Upsilon}\left(\mathbf{Z}-\frac{1}{L}\nabla \Upsilon(\mathbf{Z})\right)\right),
\end{equation}
where $\mathrm{prox}_L^{\Upsilon}(\mathbf{Z}):=\mathop{\arg\min}\limits_{\mathbf{U}\in\mathcal{Z}} \left\{\Upsilon(\mathbf{U}) +\frac{L}{2}\|\mathbf{U}-\mathbf{Z}\|^2\right\}$ is the proximal mapping associated to $\Upsilon$~\cite{J_Beck15AM}. From~\eqref{eq:gra_map_def}, the corresponding partial gradient mappings on $\mathbf{\Xi}$ with constant $L_1<\infty$ and  $\boldsymbol{\Theta}$ with constant $L_2<\infty$ are respectively given by 
\begin{equation}
G^{\mathbf{\Xi}}_{L_1}(\mathbf{Z})=L_1\left(\mathbf{\Xi}-\mathrm{prox}_{L_1}^{\Upsilon}\left(\mathbf{\Xi}-\frac{1}{L_1}\nabla_{\mathbf{\Xi}} \Upsilon(\mathbf{Z})\right)\right),
\end{equation}
\begin{equation}
G^{\boldsymbol{\Theta}}_{L_2}(\mathbf{Z})=L_2\left(\boldsymbol{\Theta}-\mathrm{prox}_{L_2}^{\Upsilon}\left(\boldsymbol{\Theta}-\frac{1}{L_2}\nabla_{\boldsymbol{\Theta}} \Upsilon(\mathbf{Z})\right)\right).
\end{equation}
Then, based on partial gradient mappings, we have the following sufficient increase property.
\begin{lemma}\label{lem:suf_inc_prop}
Updating $\mathbf{Z}$ by using AM iteration, the following inequalities hold for $ l\in\{1,2\ldots\}$
\begin{equation}\label{ineq:inter_increase_proper}
 \Upsilon(\mathbf{Z}^{(l+\frac{1}{2})})-\Upsilon(\mathbf{Z}^{(l)})\geq \frac{1}{2L_1}\| G^{\mathbf{\Xi}}_{L_1}(\mathbf{Z}^{(l)}) \|^2, 
\end{equation}
\begin{equation}\label{ineq:increase_property_theta}
\Upsilon(\mathbf{Z}^{(l+1)})-\Upsilon(\mathbf{Z}^{(l+\frac{1}{2})})\geq \frac{1}{2L_2}\| G^{\boldsymbol{\Theta}}_{L_2}(\mathbf{Z}^{(l+\frac{1}{2})})\|^2.
\end{equation}
\end{lemma}
\begin{proof}
Since the converged $\mathbf{\Xi}^{(l+1)}$ for the $l^{th}$ iteration is a local optimal point of $\mathcal{Q}1$ 
according to Theorem~\ref{lem:opt_xi_alg1}, $(\mathbf{\Xi}^{(l+1)},\boldsymbol{\Theta}^{(l)})\in \mathrm{prox}_L^{\Upsilon}(\mathbf{Z}^{(l)})$. 
Furthermore, the partial gradient of $\Upsilon(\mathbf{Z})$ is Lipschitz continuous with respect to $\mathbf{\Xi}$ for any $\boldsymbol{\Theta}$ satisfying constraint~\eqref{eq:simplex_theta}. 
Applying the result in~\cite[Lemma 2]{J_Bolte14ConverAnalysis}, we have 
$\Upsilon(\mathbf{\Xi}^{(l+1)},\boldsymbol{\Theta}^{(l)})-\Upsilon(\mathbf{\Xi}^{(l)},\boldsymbol{\Theta}^{(l)})\geq \frac{1}{2L_1}\| G^{\mathbf{\Xi}}_{L_1}(\mathbf{Z}^{(l)}) \|^2$,
which is equivalent to~\eqref{ineq:inter_increase_proper}.
On the other hand, since the converged $\boldsymbol{\Theta}^{(l+1)}$  is a stationary point of $\mathcal{D}2$, we have $(\mathbf{\Xi}^{(l+1)},\boldsymbol{\Theta}^{(l+1)})\in \mathrm{prox}_L^{\Upsilon}(\mathbf{Z}^{(l+\frac{1}{2})})$. 
Similarly, we have 
$\Upsilon(\mathbf{\Xi}^{(l+1)},\boldsymbol{\Theta}^{(l+1)})-\Upsilon(\mathbf{\Xi}^{(l+1)},\boldsymbol{\Theta}^{(l)})\geq \frac{1}{2L_2}\| G^{\boldsymbol{\Theta}}_{L_2}(\mathbf{Z}^{(l+\frac{1}{2})}) \|^2$, which is equivalent to~\eqref{ineq:increase_property_theta}.
\end{proof}

Finally, we prove the limit point $\mathbf{Z}^*$ is a stationary point of $\mathcal{P}2^{[m,j]}$ based on Lemma~\ref{lem:limit_point_property} and Lemma~\ref{lem:suf_inc_prop}.
From Lemma~\ref{lem:limit_point_property}, there exists a subsequence  $\{\mathbf{Z}^{(l)}\}_{l\in \mathbb{N}}$ converges to a limit point $\mathbf{Z}^*$, and $\{\Upsilon(\mathbf{Z}^{(l)})\}_{l\in \mathbb{N}}$ is a nondecreasing upper-bounded sequence. 
Hence, $\{\Upsilon(\mathbf{Z}^{(l)})\}_{l\in \mathbb{N}}$ must converge to some finite value and $\Upsilon(\mathbf{Z}^{(l+1)})-\Upsilon(\mathbf{Z}^{(l)})\rightarrow 0$ as $l\rightarrow \infty$. Together with Lemma~\ref{lem:suf_inc_prop}, we conclude that $G^{\mathbf{\Xi}}_{L_1}(\mathbf{Z}^{(l)})\rightarrow 0$ and $G^{\boldsymbol{\Theta}}_{L_2}(\mathbf{Z}^{(l+\frac{1}{2})})\rightarrow 0$ as $l\rightarrow \infty$, which implies that $G^{\mathbf{\Xi}}_{L_1}(\mathbf{Z}^*)= \mathbf{0}$ and $G^{\boldsymbol{\Theta}}_{L_2}(\mathbf{Z}^*)= \mathbf{0}$ due to the continuity of partial gradient mappings $G^{\mathbf{\Xi}}_{L_1}$ and $G^{\boldsymbol{\Theta}}_{L_2}$.
Therefore, the limit point $\mathbf{Z}^*$ is a stationary point~\cite{J_Bolte14ConverAnalysis}.

\bibliographystyle{IEEEtran}

\bibliography{NOMA_Secure}

\end{document}